\documentclass[reprint]{revtex4-2}

\usepackage{graphicx} 
\usepackage[inkscapelatex=false]{svg}
\usepackage[a4paper,top=2cm,bottom=2cm,left=2cm,right=2cm,marginparwidth=1.75cm]{geometry}
 
\usepackage{amsmath}
\usepackage{amssymb}
\usepackage{amsthm}
\usepackage{amsfonts}
\usepackage{graphicx}
\usepackage{subcaption}
\usepackage{xspace}
\usepackage[colorlinks=true, allcolors=blue]{hyperref}
\usepackage{enumerate}
\usepackage{color}
\usepackage{url}
\usepackage{centernot}
\usepackage{xspace}
\usepackage{hyperref}
\usepackage{enumitem}
\usepackage{adjustbox}

\newcommand{\cS}{\mathcal{S}}
\newcommand{\cE}{\mathcal{E}}

\newcommand\eps{\varepsilon}

\newcommand\Z{\mathbb{Z}}

\newcounter{marginalnote}

\usepackage{caption}
\usepackage{ragged2e}
\DeclareCaptionFormat{myformat}{\justifying#1#2#3}
\begin{document}

\title{Minimum Weight Decoding in the  Colour Code is NP-hard}

\author{Mark Walters}
\email{mark.walters@riverlane.com}
\affiliation{Riverlane, Cambridge, CB2 3BZ, UK}
\author{Mark L. Turner}
\affiliation{Riverlane, Cambridge, CB2 3BZ, UK}
\date{\today}
\begin{abstract}
All utility-scale quantum computers will require some form of Quantum
Error Correction in which logical qubits are encoded in a larger
number of physical qubits. One promising encoding is known as the
colour code which has broad applicability across all qubit types and
can decisively reduce the overhead of certain logical operations when
compared to other two-dimensional topological codes such as the
surface code. However, whereas the surface code decoding problem can
be solved exactly in polynomial time by finding minimum weight
matchings in a graph, prior to this work, it was not known whether
exact and efficient colour code decoding was possible. Optimism in
this area, stemming from the colour code's significant structure and
well understood similarities to the surface code, fanned this
uncertainty. In this paper we resolve this, proving that exact
decoding of the colour code is NP-hard -- that is, there does not
exist a polynomial time algorithm unless $\text{P}=\text{NP}$. This highlights a
notable contrast to some of the colour code's key competitors, such as
the surface code, and motivates continued work in the narrower space
of heuristic and approximate algorithms for fast, accurate and
scalable colour code decoding.
\end{abstract}

\maketitle

\newtheorem{theorem}{Theorem}
\newtheorem{lemma}[theorem]{Lemma}
\newtheorem{proposition}[theorem]{Proposition}
\newtheorem{corollary}[theorem]{Corollary}
\newtheorem*{defn}{Definition}
\newtheorem*{theorem*}{Theorem}
\newtheorem*{lemma*}{Lemma}
\newtheorem*{conjecture}{Conjecture}
\newtheorem{question}{Question}
\newtheorem*{quotedresult}{Lemma A}
\theoremstyle{remark}
\newtheorem*{remark}{Remark}

\section{Introduction}

Quantum computing has the potential to solve many problems that are
classically
intractable~\cite{babbushGrandChallengeQuantum2025}. However, the most
interesting of these will require a very large number of operations
(in excess of $10^{10}$) which means the error rate needs to be very
low. As physical qubits have much higher error rates than this requires, Quantum Error
Correction (QEC) is vital. In QEC, logical information is distributed
across physical qubits in a \emph{quantum code} such that, if the
physical error rate is below a threshold, logical errors are
exponentially suppressed in the system
size~\cite{knillResilientQuantumComputation1998}. Choosing a quantum code involves balancing
the number of physical qubits required for a given logical fidelity
with the overhead of performing quantum gates at the logical
level -- all the while meeting core hardware constraints such as qubit
connectivity.  A prominent error correction code is the surface
code~\cite{kitaevQuantumComputationsAlgorithms1997,fowlerSurfaceCodesPractical2012}, which can be implemented with
two-dimensional nearest neighbour hardware and has seen impressive experimental
demonstrations across multiple qubit types~\cite{acharyaQuantumErrorCorrection2025, cauneDemonstratingRealtimeLowlatency2024, krinnerRealizingRepeatedQuantum2022, zhaoRealizationErrorCorrectingSurface2022, bluvsteinLogicalQuantumProcessor2024, bluvsteinArchitecturalMechanismsUniversal2026}.

Colour codes~\cite{bombinTopologicalQuantumDistillation2006,
  bombinTransversalGatesError2018,
  kubicaUniversalTransversalGates2015} are closely related to surface
codes~\cite{kubicaUnfoldingColorCode2015} and are the subject of this
paper. They have also been realised
experimentally~\cite{lacroixScalingLogicColour2025,
  salesrodriguezExperimentalDemonstrationLogical2025,
  ryan-andersonRealizationRealTimeFaultTolerant2021,
  bluvsteinLogicalQuantumProcessor2024} and have memory footprints
that are competitive with the surface
code~\cite{lacroixScalingLogicColour2025,
  beniTesseractSearchBasedDecoder2025}. However, colour codes
differentiate themselves through their richer set of logical
operations~\cite{kubicaUnfoldingColorCode2015,
  kubicaUniversalTransversalGates2015,
  thomsenLowoverheadQuantumComputing2024} and utility during
non-Clifford state preparation
protocols~\cite{itogawaEfficientMagicState2025,
  gidneyMagicStateCultivation2024,
  leeLowOverheadMagicState2025}. Notably, logical~S  and Hadamard gates
are transversal in the colour code, meaning a simultaneous
single-qubit gate on all data qubits suffices -- a strong contrast to
the surface code where these gates require $O(d)$ rounds of syndrome
extraction to implement the necessary lattice surgery
protocols~\cite{geherErrorcorrectedHadamardGate2024,
  gidneyInplaceAccessSurface2024} or increased connectivity to enable
fold
transversality~\cite{moussaTransversalCliffordGates2016,chenTransversalLogicalClifford2024}.

In order to perform quantum error correction we need a
\emph{decoder}. Decoders are \emph{classical} algorithms which take
measurement data from our quantum code as input, and output a
consistent set of errors that explain these
measurements. \emph{Minimum weight decoding} involves finding the most
likely physical error consistent with the measurements. In the surface
code this can be solved exactly in polynomial time using a Minimum
Weight Perfect Matching decoder (MWPM) based on Edmond's Blossom
algorithm~\cite{edmondsPathsTreesFlowers1965}.  In contrast, while MWPM is used as a subroutine in
many colour code decoders~\cite{delfosseDecodingColorCodes2014, sahayDecoderTriangularColor2022, gidneyNewCircuitsOpen2023, kubicaEfficientColorCode2023, leeColorCodeDecoder2025, buttDecoding3DColor2025}, these decoders fail sooner than
would be expected from the code distance -- for a  code of
distance $d$ an exact minimum weight decoder will correct all errors
up to weight $\lfloor d / 2 \rfloor$, (as well as many higher weight
errors). However, all known exact decoders for the colour code have very poor worst case
runtimes~\cite{berentDecodingQuantumColor2024,beniTesseractSearchBasedDecoder2025}. Promising numerics have
been observed with heuristic decoders based on both belief
propagation~\cite{koutsioumpasColourCodesReach2025} and machine
learning~\cite{seniorScalableRealtimeNeural2025}, but these decoders make no formal
guarantees on the behaviour of error suppression at high
distances. Remarkably, whether the colour code will replace the
surface code as the leading planar architecture hinges not so much on
the code itself but whether sufficiently good decoders can be found. Prior to
this work it was unknown whether there exists an efficient and exact
minimum weight decoder for colour codes. Indeed, whilst general 
hypergraph matching is NP-hard~\cite{karpReducibilityCombinatorialProblems1972}, the prospect of a polynomial time algorithm for the colour code has remained promising because the
colour code seems qualitatively simpler than this more general
problem due to its significant structure and symmetry.

In this paper we prove that, for the colour code, there is no
polynomial time algorithm that returns the minimum weight decoding for all syndromes (assuming that
$\text{P}\not=\text{NP}$).
We prove that this is hard even in the
\emph{simplest} possible model for decoding (code-capacity with
independent, equally likely errors). This implies that that there is no
polynomial time decoding algorithm in more realistic models
such as weighted models or phenomenological noise, and strongly suggests the same for circuit level
noise.  Our main technique is to construct a very specific decoding
problem such that solving this decoding problem (i.e., finding the
most likely set of errors) is equivalent to solving the well-known
3-SAT problem (see Section~\ref{ss:3-sat}). Since 3-SAT does not have
a polynomial time solution (unless $\text{P}=\text{NP}$) this proves
our result. Our result is not intended to discourage decoder research
in the colour code space -- the colour code remains highly
compelling. Rather we aim to guide the direction of future decoder
research: progress will be made by focusing on
developing improved heuristic or approximate decoders for colour
codes that operate in the trade-off space between accuracy and speed.

\section{Main Results}
Before we state our results we need to define our terminology. As we
said in the introduction, our results apply to most noise
models, including circuit level noise, phenomenological noise, errors
with unequal probabilities, mixtures of $X$ and $Z$ errors etc. However, we
will work in the simplest possible noise model (the
code-capacity model) where errors only occur on data qubits (there are
no measurement errors), the only errors are $X$-errors and each error
occurs independently with the same probability. This is sufficient, since hardness results in
this case immediately imply the corresponding hardness results in the
more complicated models (with some minor caveats for circuit level noise depending on the exact choice of circuit -- see Section~\ref{a:other-noise} in the
Appendix for a discussion). 

Each data qubit is in some checks; since we are only considering
$X$-errors, these are $Z$-checks. A check is triggered if it meets an
odd number of errors and we call such a check a \emph{defect}. The
\emph{syndrome} $\cS$ is the state of all the checks.  We say that an
error-set $\cE$ \emph{generates} a syndrome if the error set would
give rise to that syndrome.  Recall, a minimum weight decoding is any of the 
 most likely error-sets (i.e., sets of \emph{physical} errors)
that generate a given syndrome (i.e., any of the error strings with the
smallest number of errors).

Finally, we highlight some standard definitions from computational
complexity. A decision problem (i.e., a problem with a yes/no answer) is in P if it can be solved in polynomial time; it is in NP if a solution demonstrating a `yes' answer can be verified in polynomial time.
A problem is \emph{NP-hard} if it can be used to solve any
problem in NP with at most a polynomial increase in time; it is 
\emph{NP-complete} if it is NP-hard and is itself is in NP. Throughout this paper we
will be assuming $\text{P}\not=\text{NP}$ so the reader is free to just read NP-hard
and NP-complete as saying `there does not exist a polynomial time
algorithm for this problem.'

Before defining the colour code formally (see Section~\ref{s:colour-code}) we state our results.

\begin{theorem}\label{t:min_weight}
  In the colour code, the problem of finding a minimum weight decoding for a given syndrome is NP-hard.  
\end{theorem}

The key step is the following lemma:
\begin{lemma}\label{l:given_weight}
In the colour code the decision problem of `is the  syndrome $\cS$
generated by a set of at most $k$ errors' is NP-complete.
\end{lemma}
In both of these results we prove the NP-hardness with respect to the
diameter $D$ of the syndrome, where the diameter is defined to be the
maximum distance (length of a path in the dual lattice graph) between any two nodes of
the syndrome.  Since the number of nodes in the syndrome, the minimum
number of errors generating the syndrome, and the size of the smallest
code containing the syndrome are all $O(D^2)$, this shows that
decoding is also NP-hard with respect to any of these parameters.

Theorem~\ref{t:min_weight} follows immediately from
Lemma~\ref{l:given_weight} -- if a polynomial time algorithm can find a minimum
weight decoding it can obviously be used to answer the decision
problem in Lemma~\ref{l:given_weight} in polynomial time.

Note that Theorem~\ref{t:min_weight} and Lemma~\ref{l:given_weight}
both apply without reference to any `boundary' -- we can work in the
infinite lattice, or a large finite lattice which is large enough that
we can ensure that no error meets both a boundary and any defect in
our construction.

Our second result is about the logical effect of the decoding. The
point is that for quantum error correction we do not actually need to
\emph{find} the most likely error, we just need to know whether it
changes the logical state. In the colour code the
  $Z$-logical state can be taken to be the mod 2 sum of the
  $Z$-measurements on every data qubit -- in this setting the $Z$-logical is
  flipped if an odd number of $X$-errors have occurred. This means
  that the decoder just needs to decide on the parity of the minimum
  weight decoding of a given syndrome. It turns out that even this is NP-hard.

\begin{theorem}\label{t:logical}
  Given a syndrome for the colour code, deciding whether the minimum
  weight decoding flips the logical is NP-hard.
\end{theorem}
  We prove Theorem~\ref{t:logical} in Section~\ref{s:logical} but the idea is the following.  We take a
  weight $d$ ($d$ odd = distance of code) logical operator and
  partition it into two pieces of weight $\lfloor d/2\rfloor$ and
  $\lceil d/2\rceil$ which necessarily have opposite effects on the
  logical but the same syndrome. We then combine this syndrome with
  the syndrome given by Lemma~\ref{l:given_weight} in such a way that,
  which one of these two pieces of the logical operator occurs in the
  minimum weight correction for the whole syndrome depends on the
  exact minimum weight decoding of the part of the syndrome given by
  Lemma~\ref{l:given_weight}. Therefore, deciding whether the minimum
  weight decoding flips the logical requires solving the decision
  problem in Lemma~\ref{l:given_weight}, which we have shown is
  NP-hard.

\section{Colour Code}\label{s:colour-code}

There are many varieties of colour code both in
two-dimensions and in higher dimensions~\cite{bombinTopologicalQuantumDistillation2006, bombinTransversalGatesError2018, kubicaUniversalTransversalGates2015, bombinGaugeColorCodes2015}. In this paper we focus on 
two-dimensional colour codes and, in
particular, the hexagonal code (defined below) as it is by far the most popular, and the proofs are
cleaner in this case.  We expect our results to generalise  to other
two-dimensional colour codes, and some higher dimensional colour
codes -- see Section~\ref{s:other-colour-codes}. 

The hexagonal colour code, shown in Figure~\ref{fig:bare-lattice}, can be defined as a stabiliser code as
follows. Start with a hexagonal lattice where each vertex is a
data-qubit and each face is both a weight-6 $X$ check and a weight-6
$Z$ check.  Since faces are either disjoint, identical, or share an
edge (2 qubits) these all commute. The faces can be three coloured
(conventionally red/green/blue) so that no two adjacent faces are the
same colour. An $X$ error on a data qubit triggers three $Z$-checks
(unless it is near one of the boundaries), one of each colour;
similarly a $Z$ error triggers three $X$-checks.

So far we have defined a code on the infinite lattice. To form a
finite code we restrict to a bounded region, typically a triangular or
a hexagonal region. When we do this, some checks near the
boundary only contain four data qubits -- see
Figure~\ref{fig:bare-lattice} -- and errors on the boundary will
meet two checks rather than three. (At the corners checks may
contain even fewer data qubits, and errors may only meet one check.)
For most of the paper we will work in the `lattice bulk' which can be
viewed as a portion of the infinite lattice, or of a sufficiently large
finite code, where our construction does not get close to the boundary.

For decoding we work on the dual lattice. The dual lattice has a node
corresponding to each face of the original lattice, and two nodes are
joined in the dual lattice if their corresponding faces in the
original lattice share an edge. The three-colouring of the faces of
the original lattice induces a three-colouring of the nodes of the
dual lattice.  In the dual lattice the nodes correspond to checks and
the faces correspond to data qubits. See Figures~\ref{fig:2-errors}
and~\ref{fig:3-errors} for some examples. Since each face of the dual
lattice is a triangle, we see that (as expected) each error triggers
three checks. In this paper we work entirely in the dual lattice: we
are given a syndrome (i.e., a set of defects), which is just a set of
nodes in the dual lattice. We are looking for the smallest number of
errors (i.e., faces) that generates the syndrome.

\begin{figure}
  \begin{subfigure}[t]{\columnwidth}
  \begin{subfigure}[t]{1\columnwidth}
    \begin{adjustbox}{trim=1 0 22 0, clip}
    \includegraphics[width=\columnwidth]{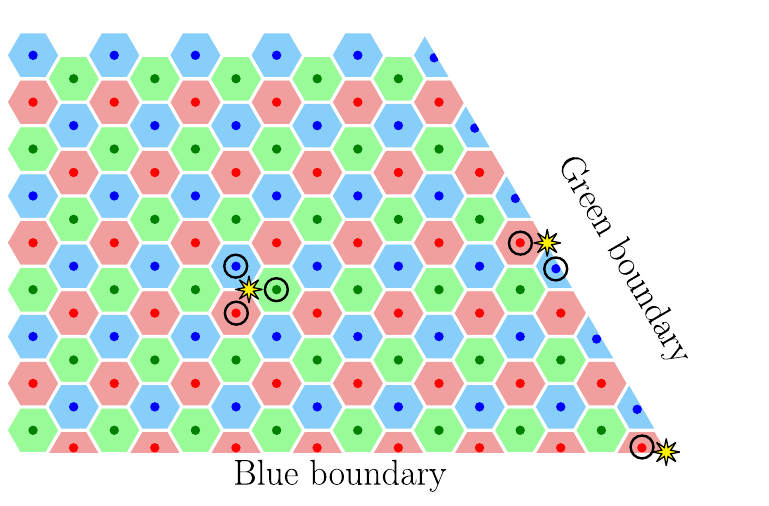}%
    \end{adjustbox}
    \caption{\label{fig:bare-lattice}}
      \begin{subfigure}[t]{0.45\columnwidth}
    \includegraphics[width=\columnwidth]{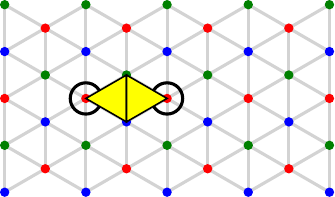}%
    \caption{\label{fig:2-errors}}
  \end{subfigure}
  \qquad 
  \begin{subfigure}[t]{0.45\columnwidth}
    \includegraphics[width=\columnwidth]{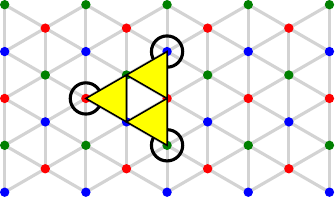}%
    \caption{\label{fig:3-errors}}
  \end{subfigure}
  \end{subfigure}
  \end{subfigure}
  \caption{\textbf{The colour code and its error model}. \textbf{(a)}
    The $6.6.6$ colour code which can be defined with data qubits on
    the vertices of a hexagonal lattice and checks on the coloured
    faces.  A Pauli error on a bulk data qubit (the yellow star in the
    centre of the picture) anti-commutes with all three stabiliser
    checks acting on that qubit, triggering the three checks ringed in
    black. However, if the data qubit lies on a boundary, then only the two checks it participates in are
    triggered (the yellow star on the middle right); and if it lies at a corner, it triggers the single check that it is contained in (the yellow star in the bottom right). We have not shown the other boundary (the red boundary) as our constructions will need rather large codes.  The coloured nodes in
    the centres of the hexagons are the vertices of the dual lattice,
    which is a representation of the colour code's error model where
    the nodes correspond to checks and the data qubit errors
    correspond to faces. \textbf{(b)} and \textbf{(c)} show errors (the yellow triangles) on the dual lattice of the colour code.  The check nodes fire (ringed in black) if they meet an
    odd number of errors so these errors partially cancel leaving the
    defect patterns shown.}
\end{figure}

\section{Outline of the proof}\label{s:proof-outline}

\subsection{Exact Covers and Separated Syndromes}\label{s:exact-cover}
We defer the formal definition of 3-SAT to Section~\ref{ss:3-sat}. For
the moment all we need is that 3-SAT is the question of deciding
whether a boolean formula~$F$ is satisfiable. Given any 3-SAT formula
$F$, we are going to construct a syndrome~$\cS$ in the colour code
with two properties.
\begin{itemize}
\item The number of errors in any error-set $\cE$ generating $\cS$ is
  at least the number of defects (i.e., $|\cE|\ge |\cS|$).
\item There exists an error set $\cE$ generating $\cS$ with $|\cE|=
  |\cS|$ if and only if the formula $F$ is satisfiable.
\end{itemize}

Suppose that we can do this. Then any algorithm that can answer the
question `is there an error set $\cE$ generating $\cS$ with $|\cE|\le
|\cS|$?' can answer the question `is the 3-SAT formula $F$
satisfiable?' In other words any algorithm that can solve the decision
problem in Lemma~\ref{l:given_weight} can be used to solve
3-SAT. Since 3-SAT is NP-complete this would prove
Lemma~\ref{l:given_weight} and thus Theorem~\ref{t:min_weight} showing
that the minimum weight decoding problem is NP-hard.

We make two definitions corresponding to the two properties
above.
\begin{defn}
  A syndrome is \emph{separated} if it does not contain any pair of
  defects at distance one (measured on the dual lattice graph) from each other.
\end{defn}
\begin{defn}
  An error-set $\cE$ is an \emph{exact cover} of a syndrome $\cS$ if
  the number of errors in $\cE$ is equal to the number of defects in $\cS$
   (i.e., $|\cE|= |\cS|$) and $\cE$ generates the syndrome
  $\cS$.
\end{defn}

Observe that, in any separated syndrome, each error contains at most
one defect. Since each defect must be contained in at least one error,
we see that we must have $|\cE|\ge |\cS|$.

In fact, this argument tells us more: in any exact cover of a separated
syndrome every defect must be contained in exactly one error, and
every error must contain exactly one defect.  These two properties
mean that an exact cover of a separated syndrome is very constrained
-- a large amount of what occurs is forced. Moreover, it is `locally'
constrained -- in an exact cover of a separated syndrome we cannot
trade off 
doing `worse' (using more errors)
in one region and make up for it by 
doing `better' (using fewer errors)
elsewhere -- the exact cover must be optimal everywhere. This
makes reasoning about the cover, and therefore the proof, much simpler
than in a more general setting. Let us give an example.  Fix a
particular error in an exact cover; one node is a defect, but the
other two are not, so must be cancelled by other
errors. Figure~\ref{fig:2-errors} and~\ref{fig:3-errors} show two
examples of how this can happen. Whilst those are not the only ways
the cancellation can occur, in many of the cases we construct it will
be easy to see that these are the only two options -- in fact in the
majority of the cases nodes will be `paired' or \emph{matched} as in
Figure~\ref{fig:2-errors}.

Finally,  there is a parity constraint on the syndrome. Given an
error-set that does not meet the boundary, since each error contains
one node of each colour, the parity of the number of
defects of each colour must be the same -- they must all be odd or all
even. It will be useful to have a `local' version of the same idea.
Given an error-set define an error component to be a connected set of
errors (where two errors are connected if they share a vertex in the
dual lattice). 
We see that the parity of the number of defects of each
colour generated by this error component must be the same. For
example, in~\ref{fig:2-errors} there are two blue defects and no red
or green defects (even parity case), whereas in~\ref{fig:3-errors}
there is one defect of each colour (odd parity case). 

We remark that, whilst asking whether a separated syndrome has an
exact cover is a very special case of the decoding problem, in
complexity theory to establish NP-hardness it suffices to show that
\emph{some} problem instances are hard to solve.  Indeed, although many
instances of 3-SAT are easy to solve,  it is still NP-complete --
there exist instances that are hard to solve. 

\begin{figure}
    \includegraphics[width=\columnwidth]{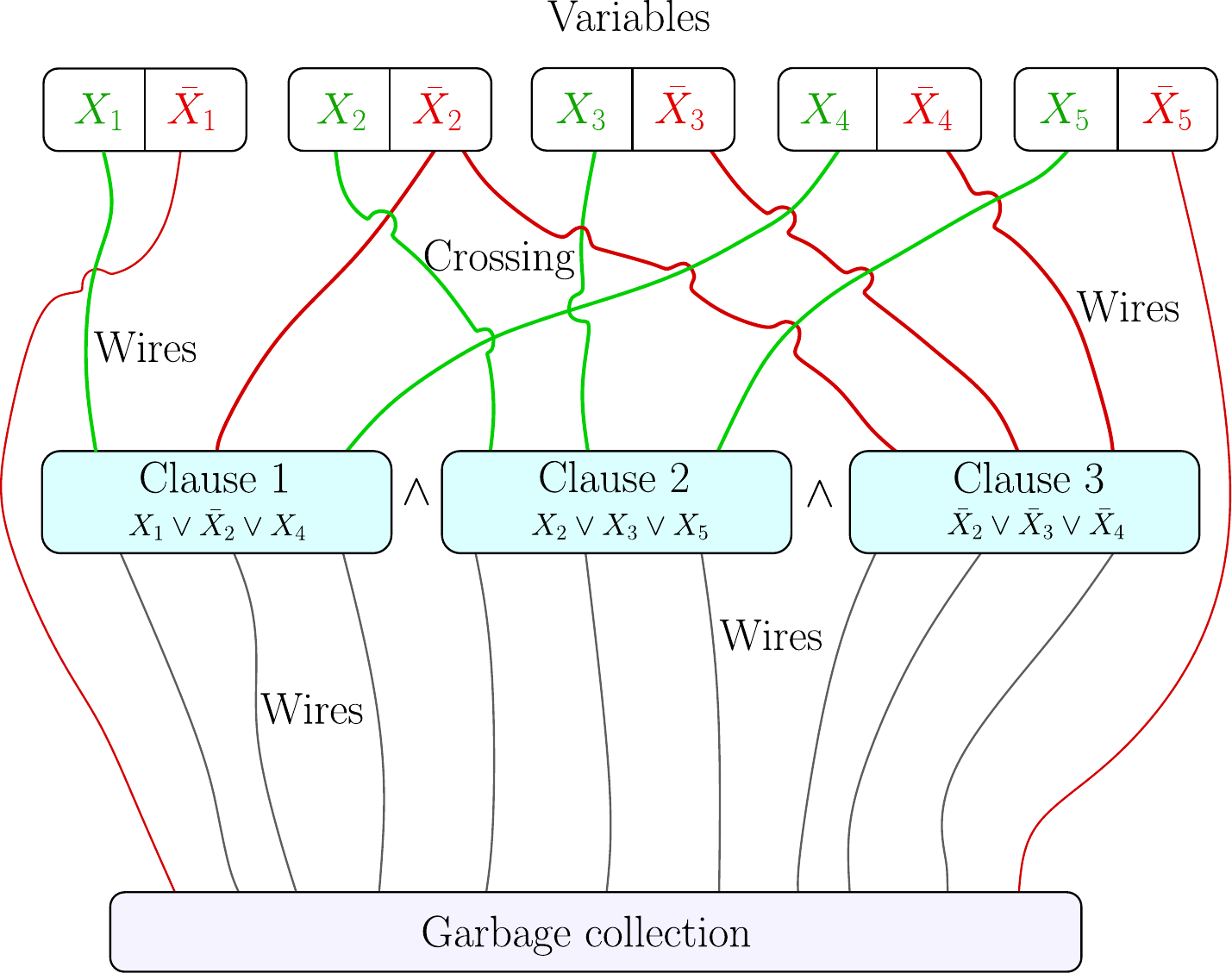}%
    \caption{\label{fig:construction} \textbf{Outline of the syndrome construction (see
      Section~\ref{s:proof-outline}).} In order to prove the NP-hardness of decoding in the colour code, we construct a specific syndrome that represents a 3-SAT problem and is, therefore, hard to decode. This figure illustrates the construction for the formula $(X_1\vee
      \bar{X}_2\vee X_4)\wedge (X_2\vee X_3\vee X_5)\wedge
      (\bar{X}_2\vee\bar{X}_3\vee \bar{X}_4)$. The boxes represent the main
      gadgets (which are collections of defects) -- for simplicity we
      have used lines to represent the wires.
      If the constructed syndrome can be covered exactly (see
      Section~\ref{s:exact-cover}), then each variable gadget must be
      in the TRUE or FALSE state, and each clause gadget must have at
      least one TRUE input, which means that the assignment is a
      satisfying assignment for the 3-SAT formula. The converse also
      holds: if the formula is satisfiable then putting the variable
      gadgets in the corresponding state allows us to exactly cover
      all the variable gadgets, wires and clause gadgets (since a
      clause gadget can be exactly covered providing at least one of
      its inputs is TRUE).
      Finally, we use the technical properties of the garbage
      collection gadget to complete the cover to an exact cover of the
      whole syndrome. }
\end{figure}

\begin{figure*}[t]
  \begin{subfigure}{0.3\textwidth}
  \includegraphics[width=\columnwidth]{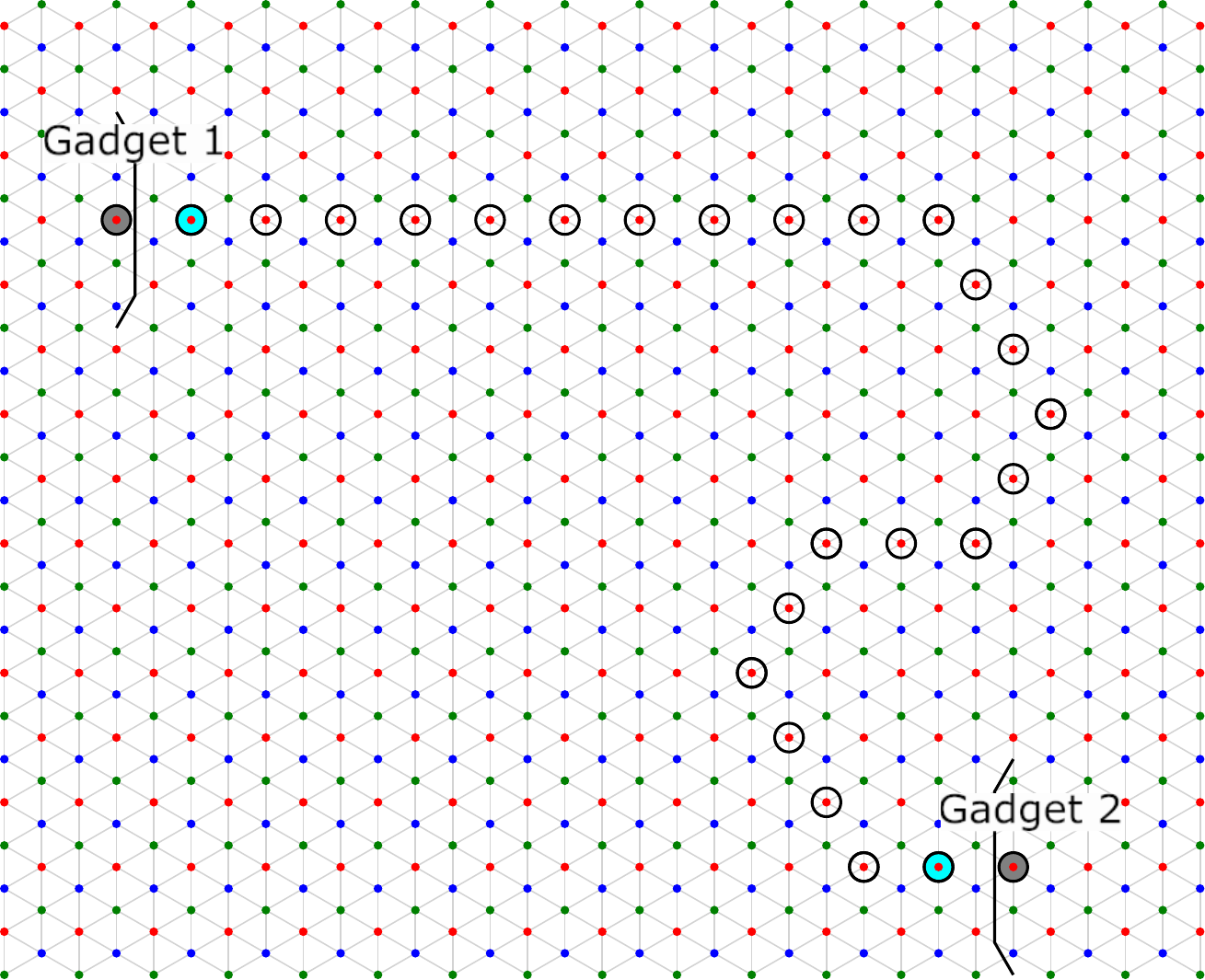}
      \caption{\centering\label{fig:wire-empty}Empty gadget}
  \end{subfigure}
  \hspace{0.03\columnwidth}
  \begin{subfigure}{0.3\textwidth}
  \includegraphics[width=\columnwidth]{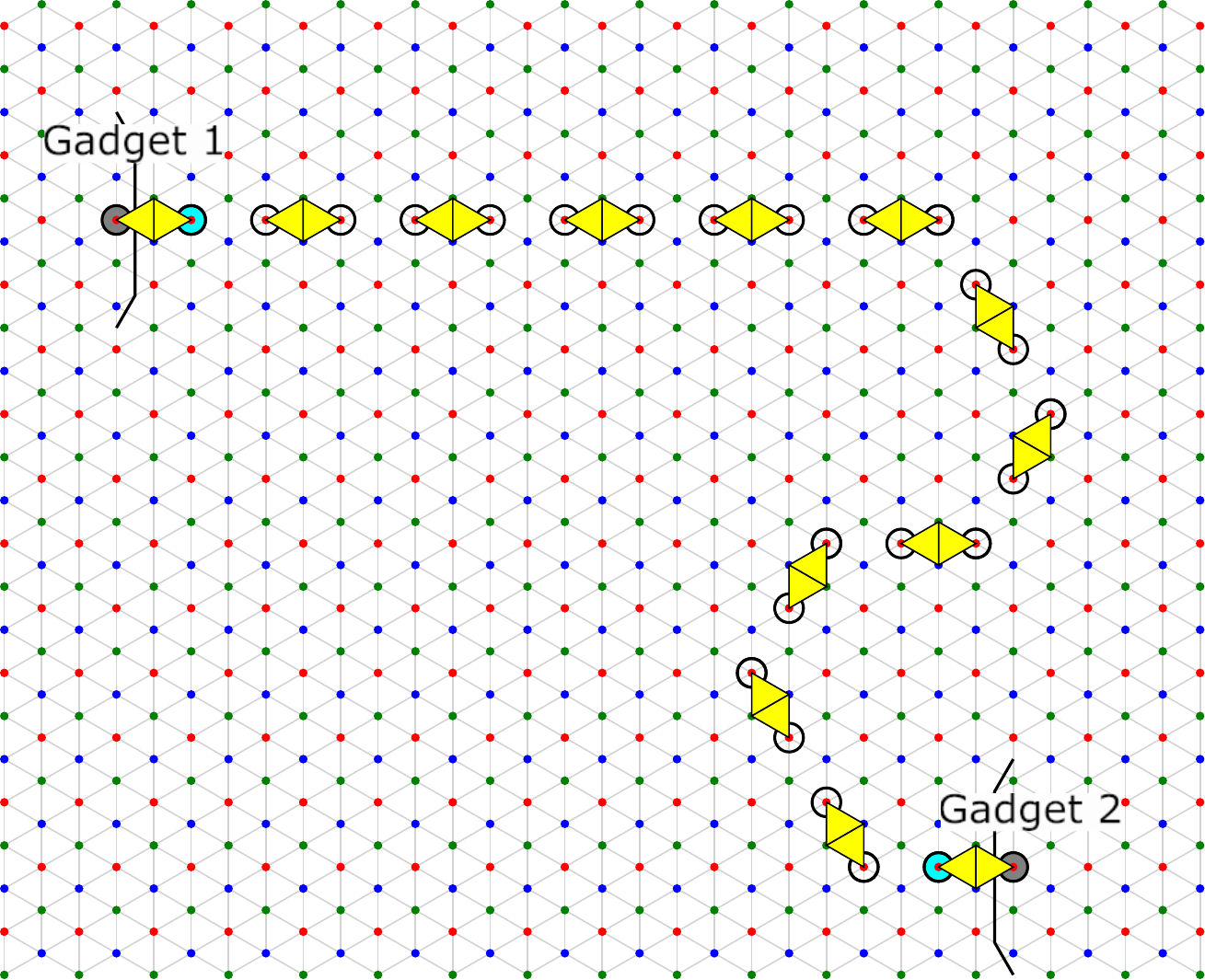}
      \caption{\centering\label{fig:wire-true}Both link nodes TRUE}
  \end{subfigure}
  \hspace{0.03\columnwidth}
  \begin{subfigure}{0.3\textwidth}
  \includegraphics[width=\columnwidth]{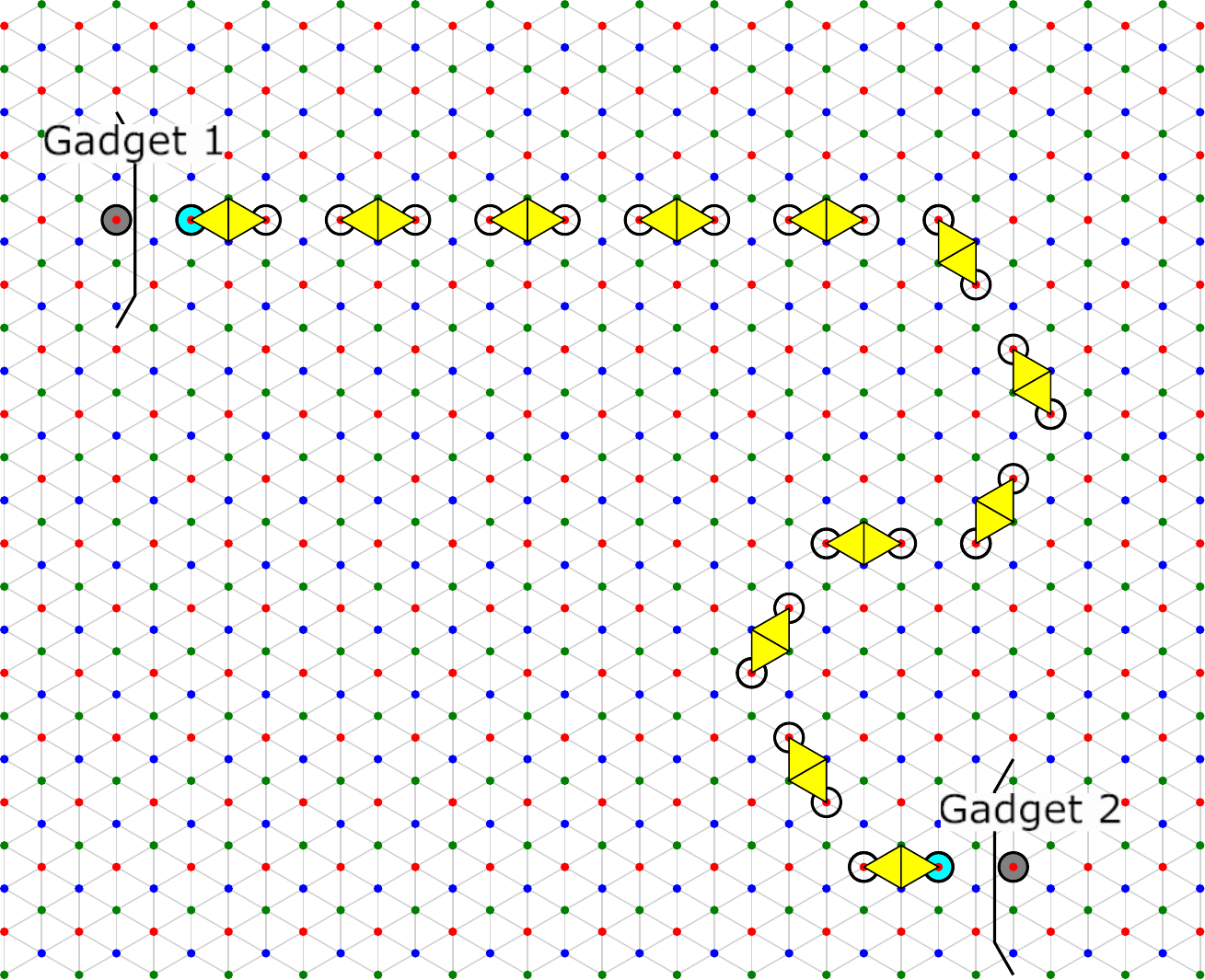}
      \caption{\centering\label{fig:wire-false}Both link nodes FALSE}
  \end{subfigure}
  \caption{\label{fig:wire-whole} \textbf{An example of the wire
      gadget}. \textbf{(a)} The nodes circled in black are defects; the cyan discs
    are the wire gadget's link nodes through which the gadget is
    connected to other gadgets; the grey discs are the partner
    link nodes in the neighbouring gadgets. The key point is that a wire
    has a `path like' structure, and has an even number of nodes. We
    see that the state of the link node in the top left is copied to
    the state of the link node in the bottom right -- TRUE to TRUE and
    FALSE to FALSE in \textbf{(b)} and \textbf{(c)} respectively.}
\end{figure*}

\subsection{3-SAT}\label{ss:3-sat}
To prove Lemma~\ref{l:given_weight} we show that any algorithm that
can solve this problem can also solve the NP-complete problem
3-SAT. For the reader's convenience, and to set notation, we recall
the definition here.

  Let $X_1,X_2,\dots, X_n$ be boolean variables each taking a value
  TRUE or FALSE.  A \emph{3-SAT formula} is an `AND' (conjunction) of
  clauses where each clause is an `OR' (disjunction) of three literals
  (a literal is a variable $X_i$ or its negation $\bar{X_i}$). For
  example
      \[(X_1  \vee \bar X_4 \vee X_6)\wedge(\bar X_3 \vee X_4 \vee \bar X_7)\wedge(\bar X_2 \vee \bar X_3 \vee \bar X_7)    \] is a 3-SAT formula with three clauses.
     The 3-SAT problem is to decide whether there is an assignment of
     TRUE/FALSE to the variables $X_i$ such that a given 3-SAT formula
     is true. 
     
     Intuitively, the reason 3-SAT is hard (while 2-SAT is not) is that, setting a variable to TRUE or FALSE does not immediately constrain any other variables, so the search space is exponential.

\subsection{Construction of Syndrome}

We have several `gadgets' that we use to construct the syndrome
$\cS$. Each gadget is a collection of defects. These gadgets provide a
simple algorithmic way to turn a 3-SAT formula into a syndrome. The
gadgets include `variable gadgets' (one for each variable), and
`clause gadgets' (one for each clause), but there are several
others. For an overview of how these gadgets are put together see
Figure~\ref{fig:construction}. 

Once we have given this construction we will prove that the
constructed syndrome has an exact cover if and only if the 3-SAT
formula is satisfiable.

We will insist that the gadgets are spaced out in such way that we can
control their interaction. When defining each gadget we will designate
a small number of its defects as \emph{link nodes}. The idea is that
these are the only places where the defects from different gadgets are
close enough to interact.  We will insist that defects in different
gadgets are at distance at least four \emph{unless} both nodes are
link nodes. Further, we will also insist that all link nodes are red,
that each link node has exactly one node not in its own gadget at
distance less than four (which must be a link node in another gadget),
and that node is at distance two. We call this node the link node's
\emph{partner}. This means that each link node is either matched to
its partner (their errors share two nodes), or their errors share no
nodes (in which case those nodes must be cancelled by errors inside
the gadget).  We will say a link node is in the \emph{TRUE} state if
it is matched to its partner and \emph{FALSE} otherwise.

In the remainder of this section we outline the properties of each of
our gadgets, but we defer the actual construction of the gadgets, and
the proof that they have the required properties, until
later. However, we do give references to the figures showing each
gadget and, in some cases, the subgadgets used to construct the
gadget. The notation for the figures is as follows: circled nodes
represent defects; circled nodes in cyan are link nodes in the gadget;
circled nodes in grey are the partner link nodes (not in the
gadget). Yellow triangles represent the errors. 

\smallskip
  \textbf{The wire gadget.} The wire gadget has two link nodes. The key property is that, in any
  exact cover, these two link nodes are in the same state: they are
  either both TRUE or both FALSE. Wires allow us to `transmit'
  information around our network. The gadget is shown in
  Figure~\ref{fig:wire-whole}.

     \begin{figure*}
  \begin{subfigure}{0.4\textwidth}
  \includegraphics[width=\textwidth]{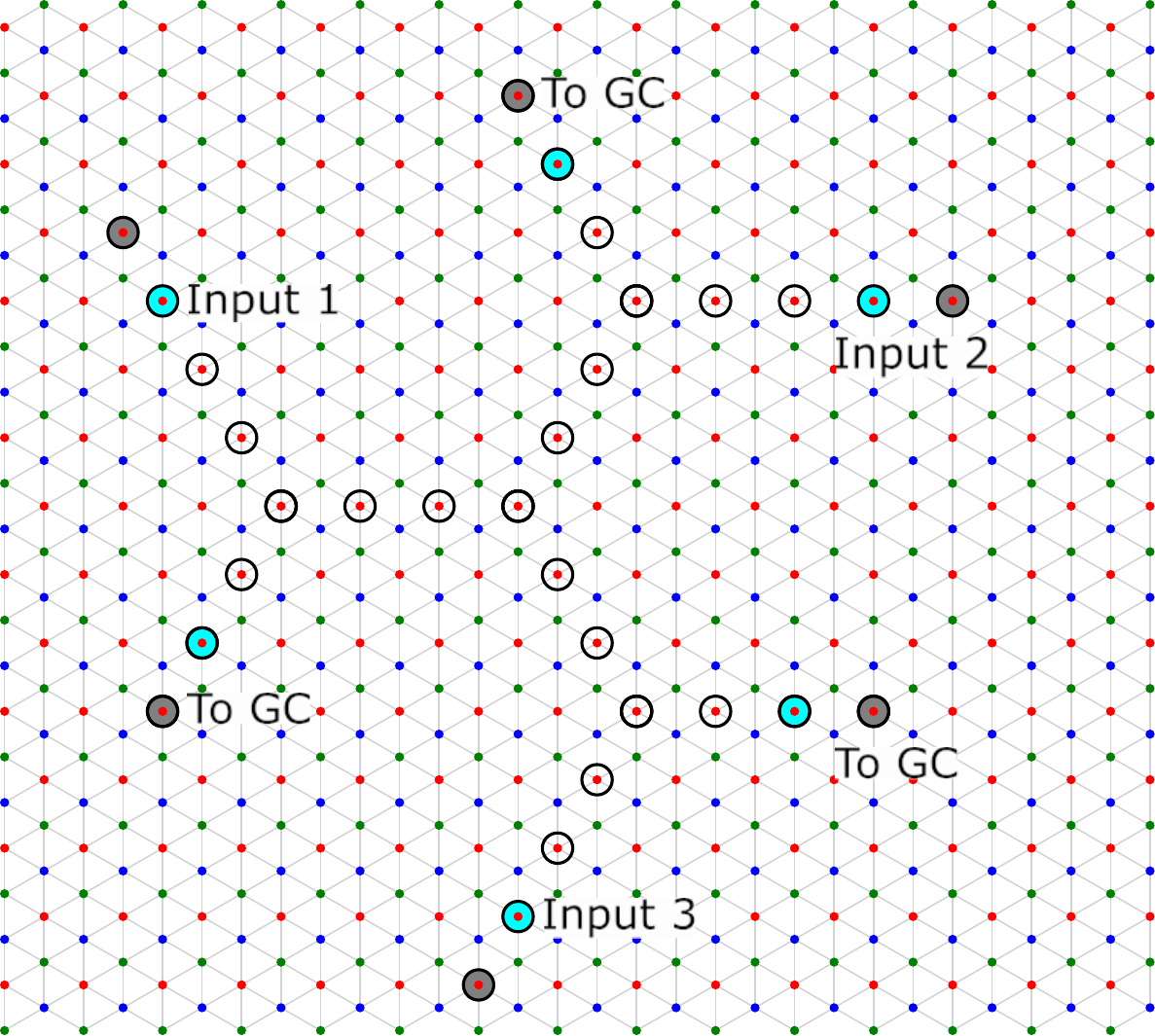}
      \caption{\centering\label{fig:clause-empty}Empty clause gadget}
  \end{subfigure}
  \hspace{0.03\textwidth}
  \begin{subfigure}{0.4\textwidth}
  \includegraphics[width=\textwidth]{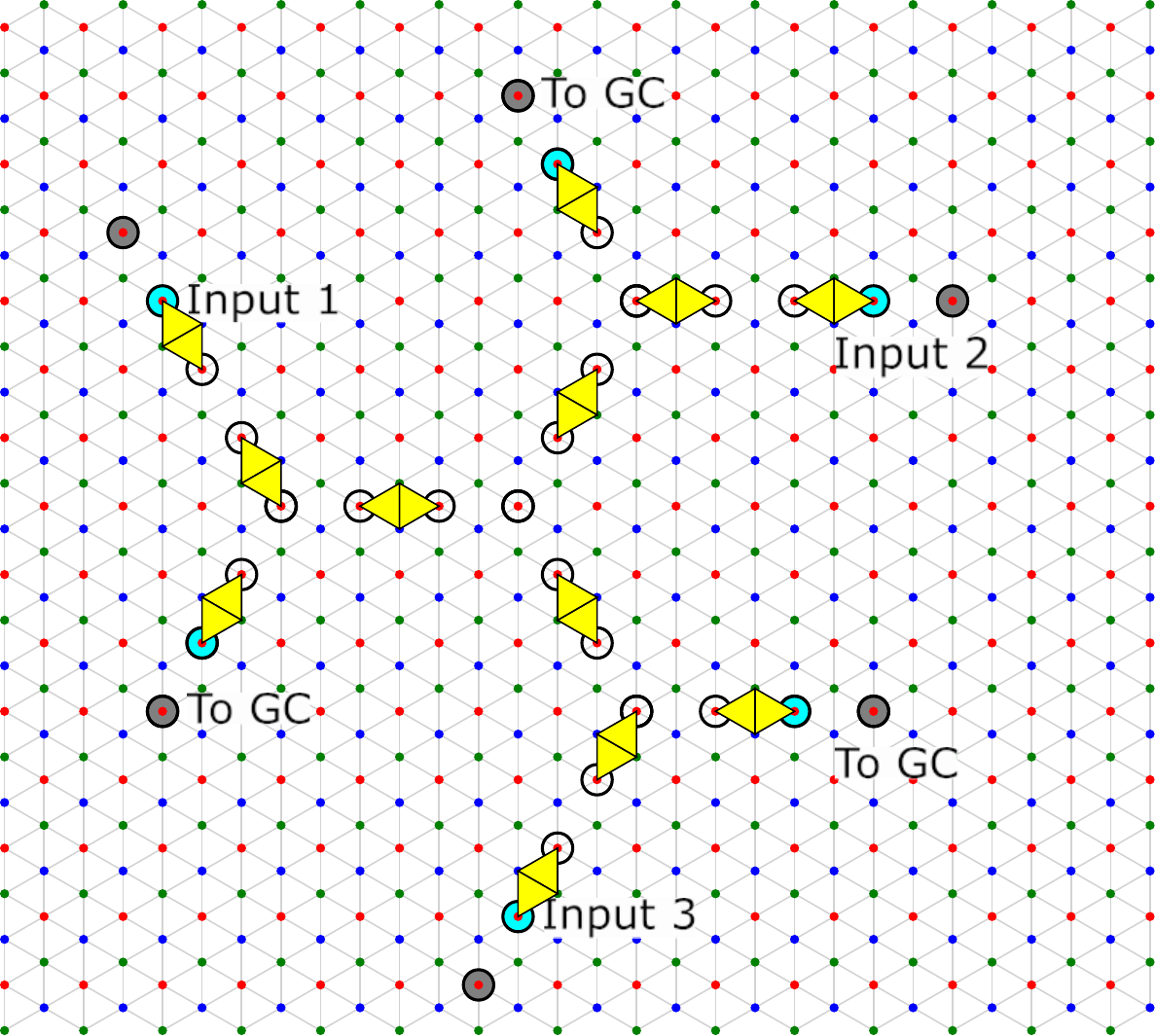}
      \caption{\centering\label{fig:clause-all-false}All inputs to the clause are FALSE}
  \end{subfigure}
  \hspace{0.03\textwidth}
  \begin{subfigure}{0.3\textwidth}
  \includegraphics[width=\textwidth]{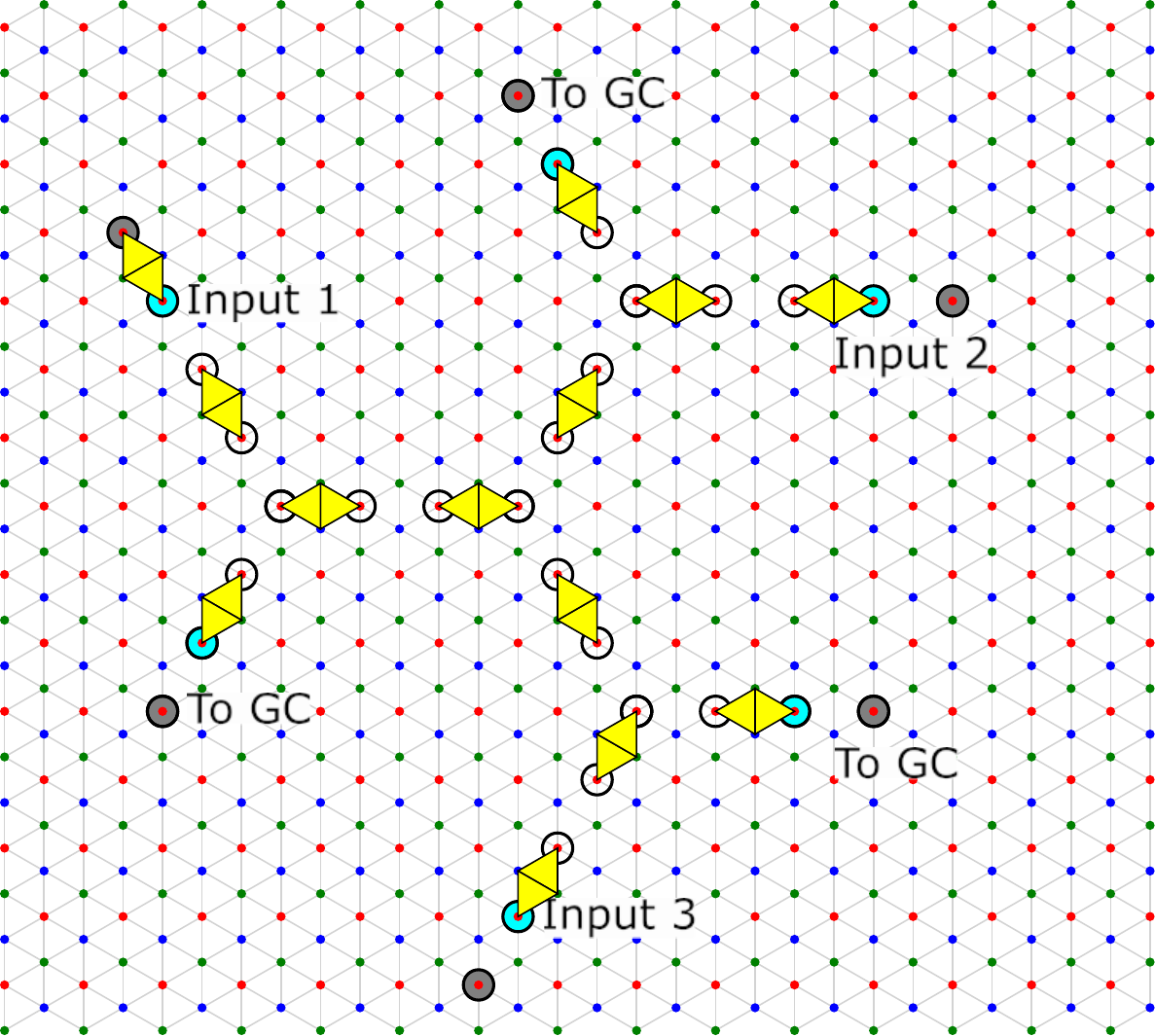}
      \caption{\centering\label{fig:clause-1-true}One input is TRUE}
  \end{subfigure}
    \hspace{0.03\textwidth}
  \begin{subfigure}{0.3\textwidth}
  \includegraphics[width=\textwidth]{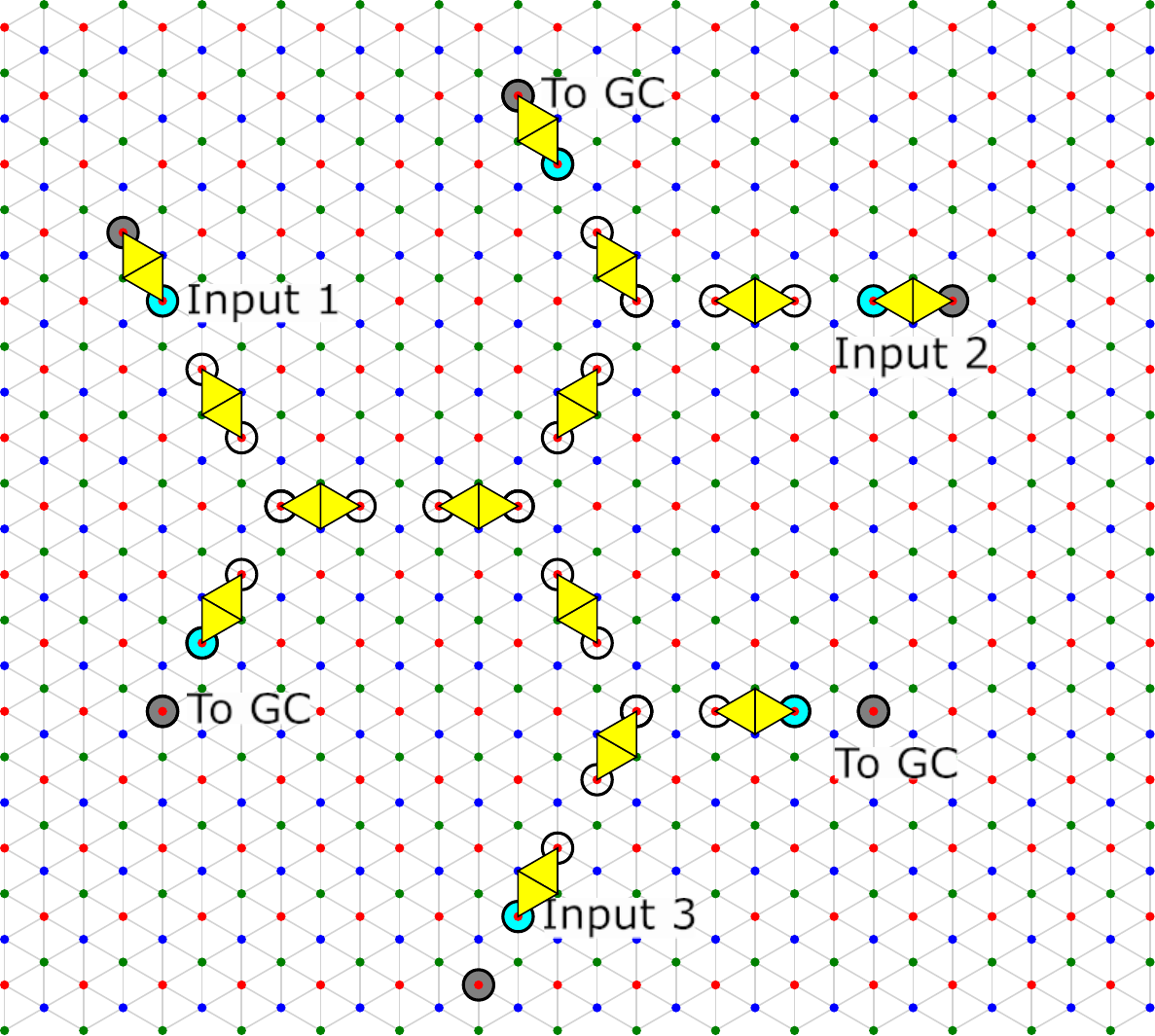}
      \caption{\centering\label{fig:clause-2-true}Two inputs are TRUE}
  \end{subfigure}
  \hspace{0.03\textwidth}
  \begin{subfigure}{0.3\textwidth}
  \includegraphics[width=\textwidth]{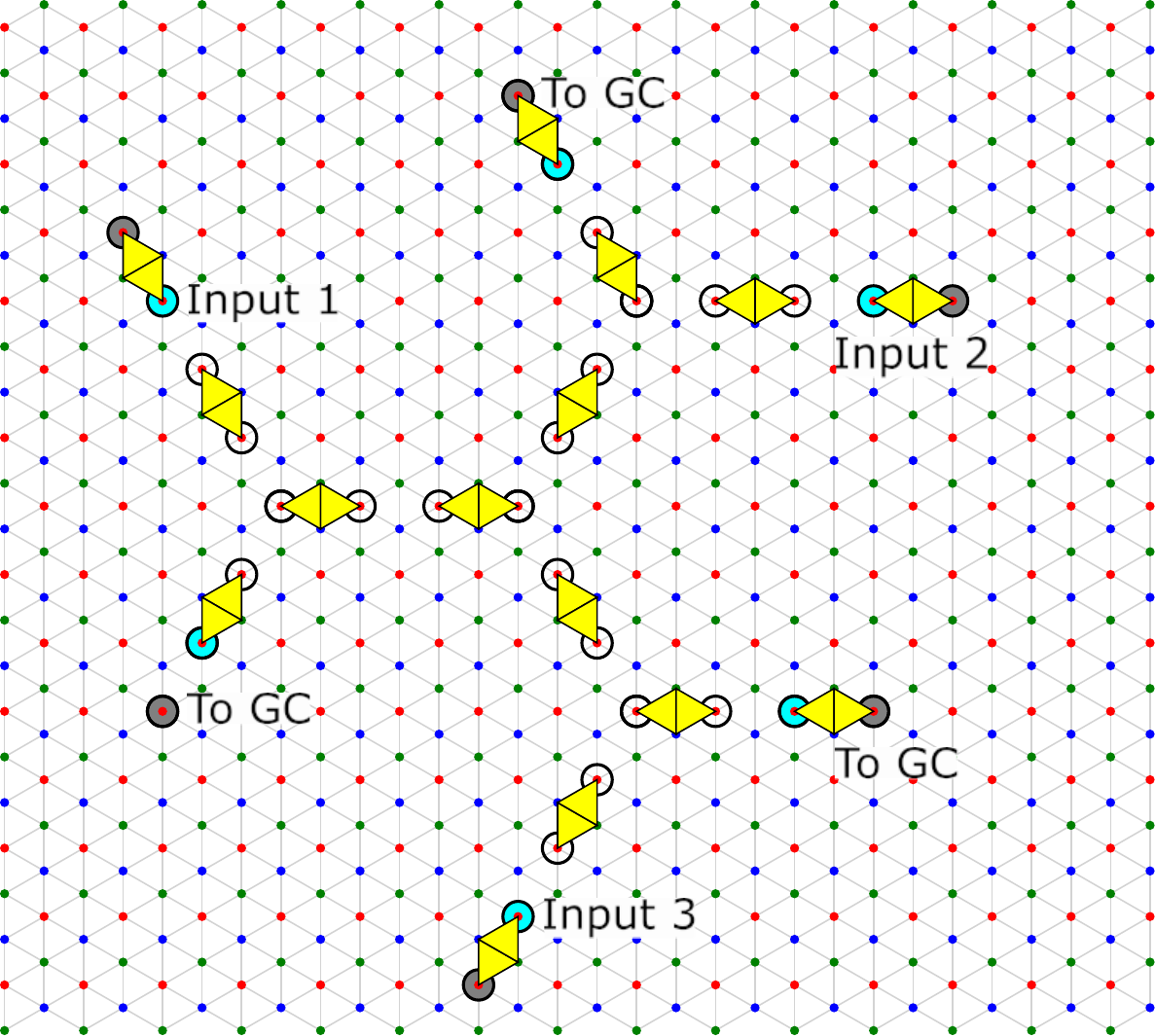}
      \caption{\centering\label{fig:clause-3-true}Three inputs are TRUE}
  \end{subfigure}
  \caption{\label{fig:clause-whole} \textbf{The clause
      gadget}. \textbf{(a)} To model 3-SAT, we want the clause gadget to have
    an exact cover if and only if at least one input is TRUE (the link
    node matched to its partner). The gadget has three input links and three links that are connected to the garbage collection gadget (GC). In \textbf{(b)} we see that the
    gadget cannot be covered exactly if all three inputs are FALSE --
    the errors show what is forced by the three inputs being
    FALSE, and we see the central node is uncovered. In contrast,
    \textbf{(c)}, \textbf{(d)} and \textbf{(e)} show that the gadget
    can be covered exactly if one, two or three inputs are TRUE. Since the gadget is rotationally symmetric, the other cases are just rotations of the cover shown in each figure.}
\end{figure*}

  \smallskip
  \textbf{The clause gadget.} The clause gadget has three `input' link
  nodes.  It has an exact cover if and only if at least one input is
  in the TRUE state.  In order to ensure the existence of an exact
  cover in each of the cases when one, two or three inputs are TRUE, we
  add three output link nodes (one corresponding to each input link
  node) to the gadget, then connect these to wires which, in turn, link
  to the garbage collection gadget which will `clean up' these
  cases. The gadget is shown in Figure~\ref{fig:clause-whole}.

  \smallskip
  \textbf{The variable gadget.}
 The variable gadget has two possible exact covers which we call TRUE
 and FALSE.  It has several `output' link nodes (we can choose how many),
 divided into two sets $A$ and $B$.
 \begin{itemize}
  \item If the gadget is in the TRUE state all outputs in $A$ are TRUE, and all outputs in $B$ are FALSE.
  \item If the gadget is in the FALSE state all outputs in $A$ are FALSE, and all outputs in $B$ are TRUE.
 \end{itemize}
  Suppose the gadget corresponds to the variable $X_i$. Then the
  outputs in $A$ have state $X_i$ and the outputs in $B$ have state
  $\bar X_i$. This gadget is relatively complicated and is shown in Figure~\ref{fig:variable-whole}.

\begin{figure*}
    \def\figbasename{rgb-duplicator}
  \begin{subfigure}{0.3\textwidth}
  \includegraphics[width=\textwidth]{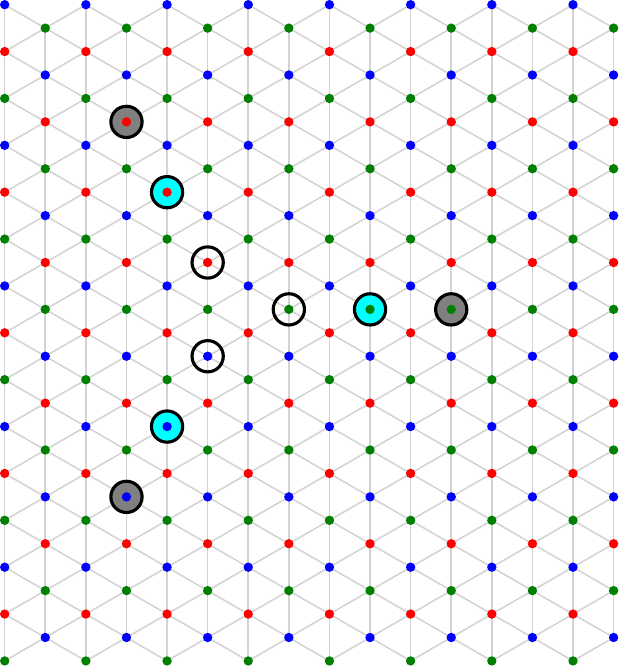}
      \caption{\centering\label{fig:rgb-subgadget-empty}Empty RGB-duplicator subgadget}
  \end{subfigure}
  \hspace{0.03\textwidth}
  \begin{subfigure}{0.3\textwidth}
  \includegraphics[width=\textwidth]{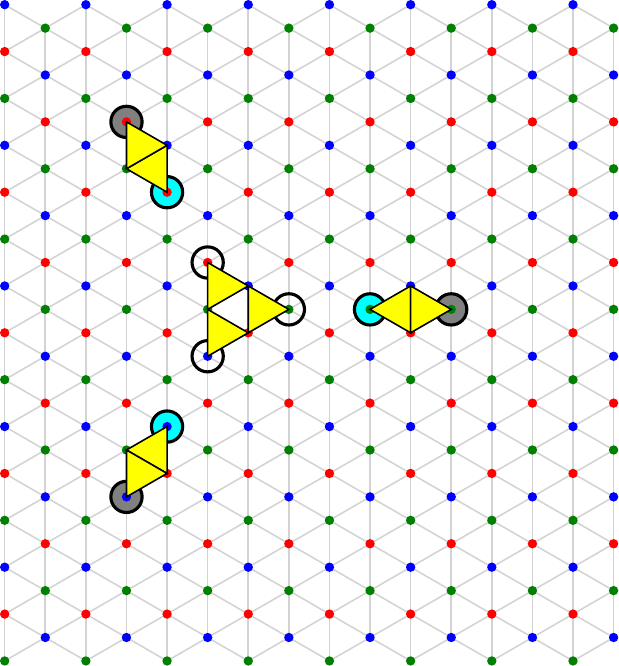}
      \caption{\centering\label{fig:rgb-subgadget-true}Link nodes TRUE}
  \end{subfigure}
  \hspace{0.03\textwidth}
  \begin{subfigure}{0.3\textwidth}
  \includegraphics[width=\textwidth]{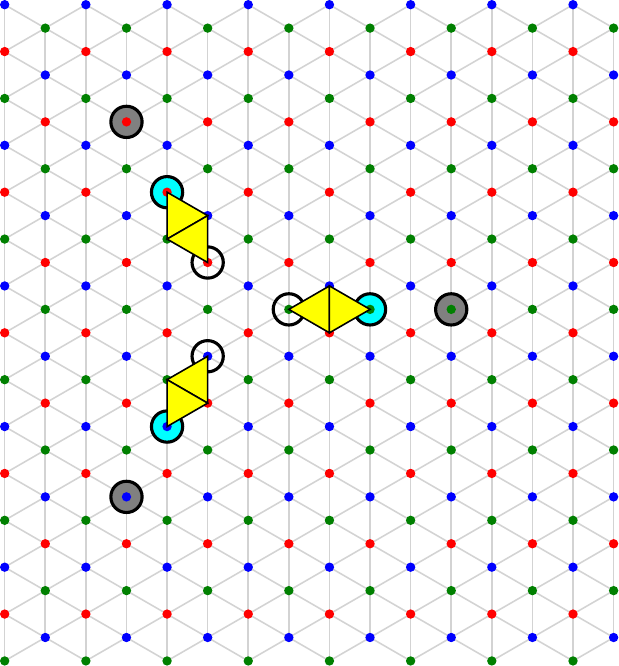}
      \caption{\centering\label{fig:rgb-subgadget-false}Link nodes FALSE}
  \end{subfigure}
    \def\figbasename{double-duplicator}
  \begin{subfigure}{0.3\textwidth}
  \includegraphics[width=\textwidth]{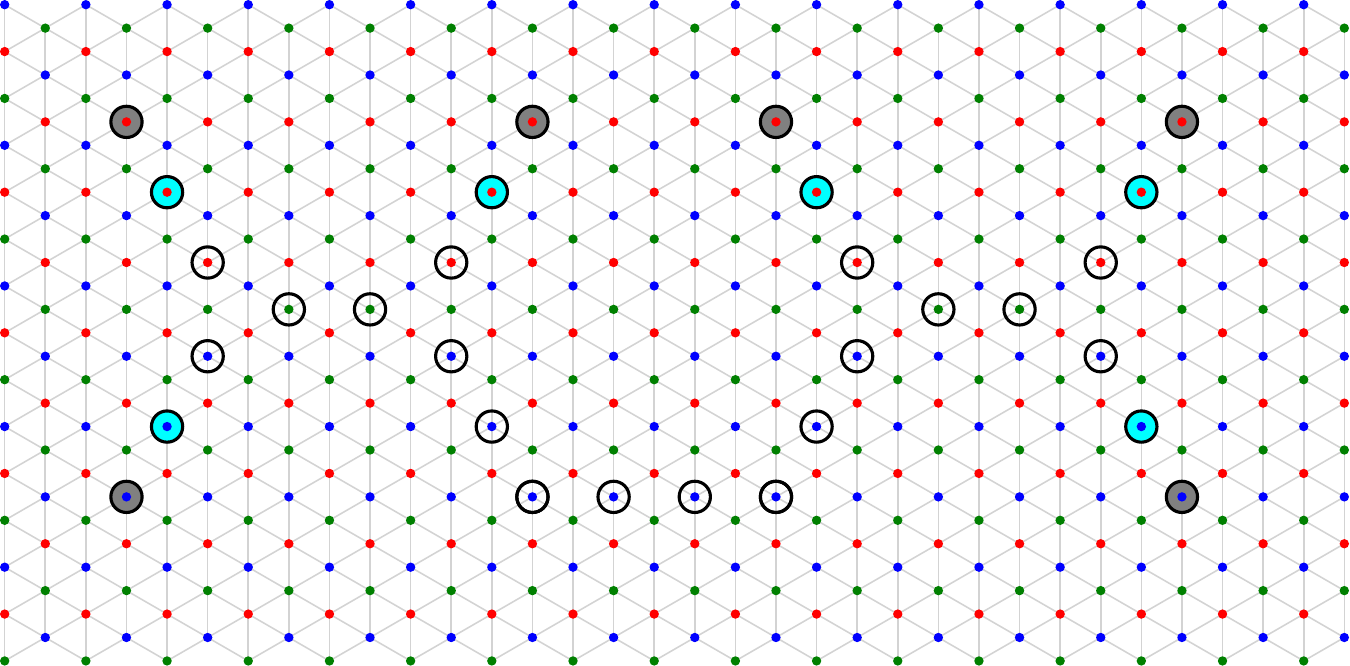}
      \caption{\centering\label{fig:double-duplicator-empty}Empty duplicator subgadget}
  \end{subfigure}
  \hspace{0.03\textwidth}
  \begin{subfigure}{0.3\textwidth}
  \includegraphics[width=\textwidth]{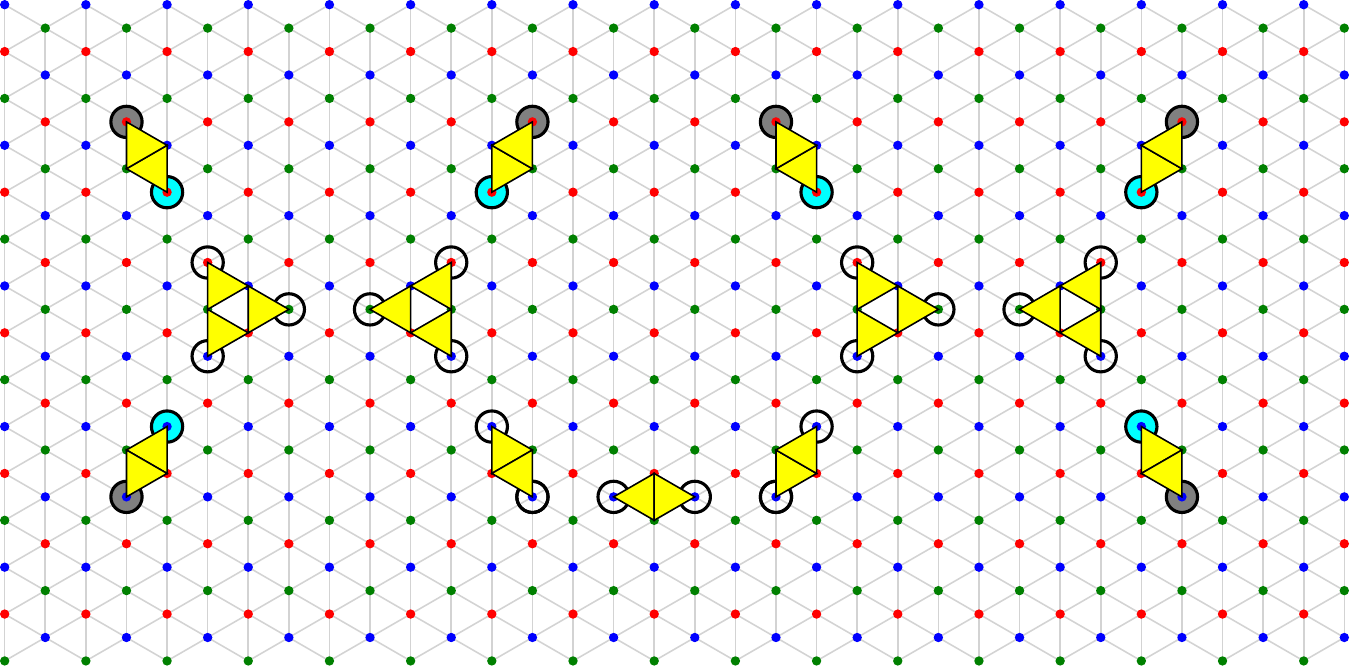}
      \caption{\centering\label{fig:double-duplicator-true}Link nodes TRUE}
  \end{subfigure}
  \hspace{0.03\textwidth}
  \begin{subfigure}{0.3\textwidth}
  \includegraphics[width=\textwidth]{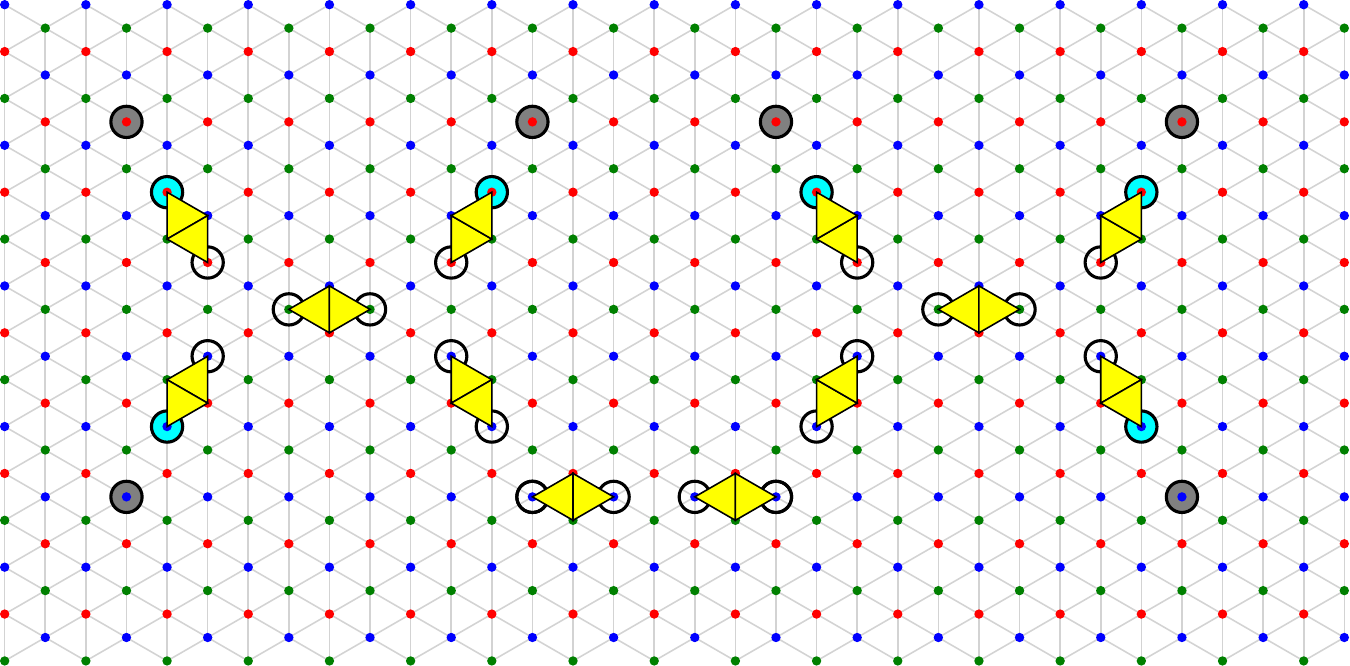}
      \caption{\centering\label{fig:double-duplicator-false}Link nodes FALSE}
  \end{subfigure}

  \begin{subfigure}{0.98\textwidth}
    \includegraphics[width=\textwidth]{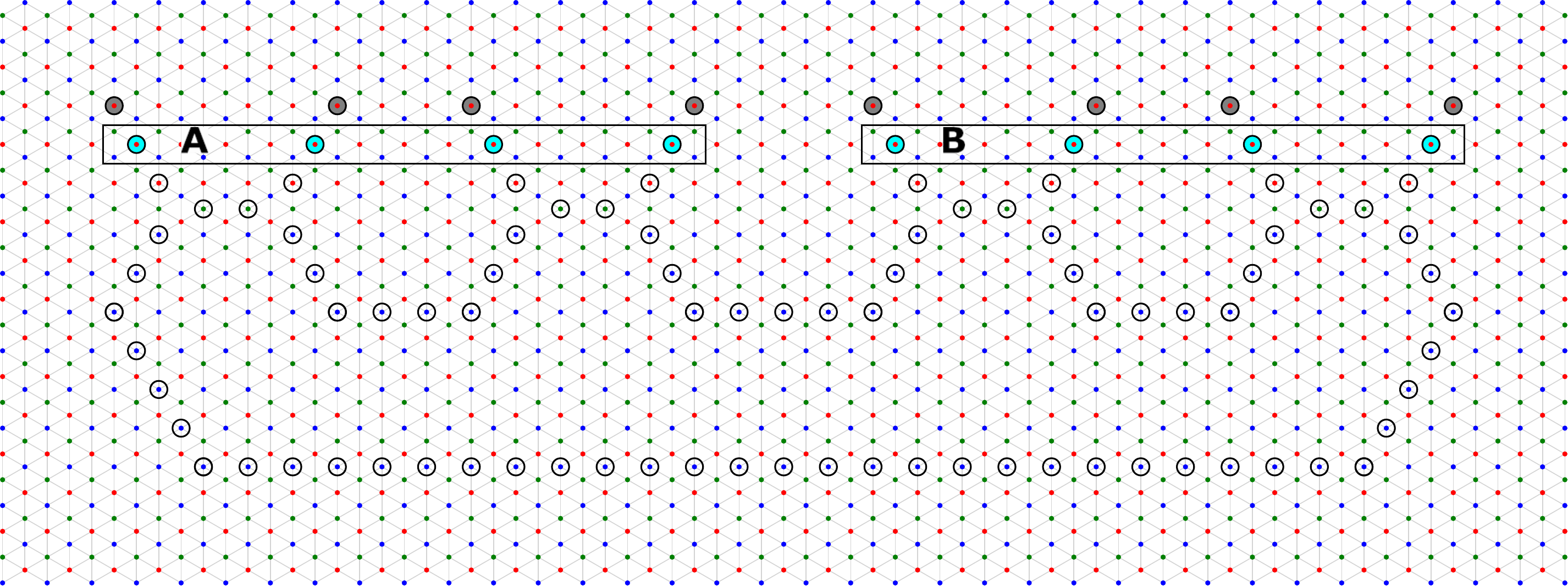}
    \caption{\centering\label{fig:variable-empty}Empty variable gadget} 
  \end{subfigure}
  \begin{subfigure}{0.98\textwidth}
    \includegraphics[width=\textwidth]{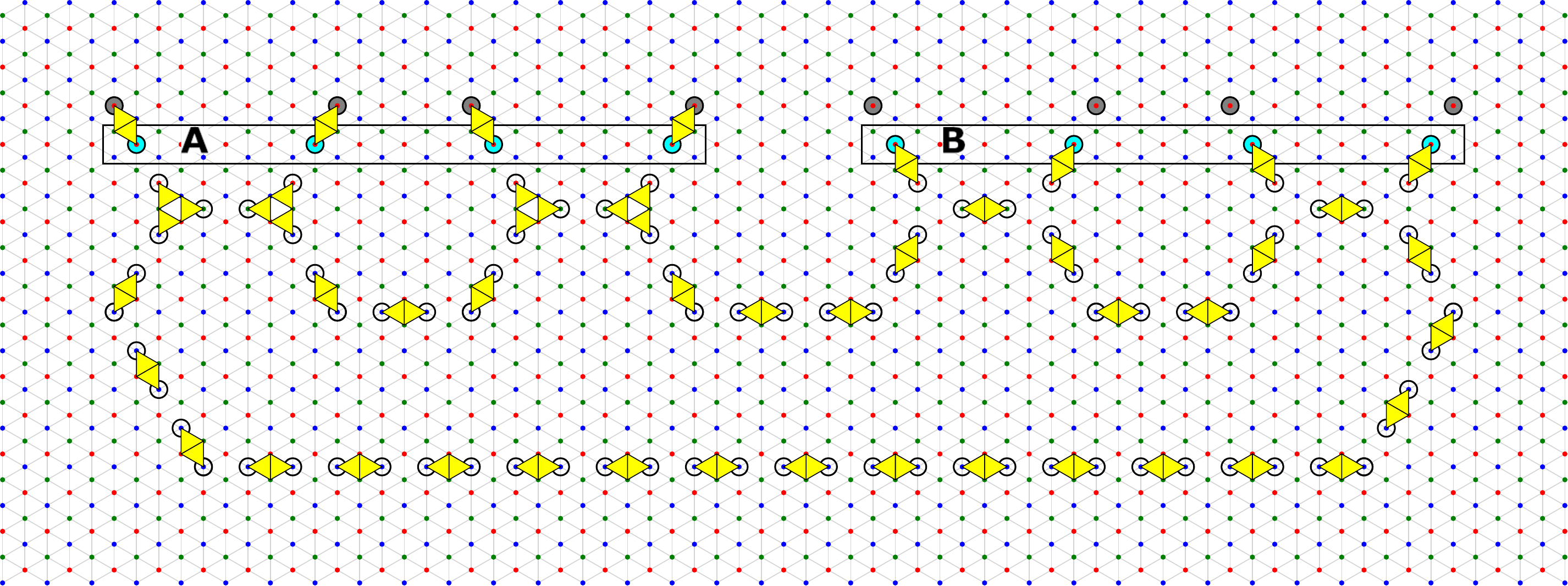} 
      \caption{\centering\label{fig:variable-true}Variable in the TRUE state (link nodes in $A$ TRUE; link nodes in $B$ false)}
  \end{subfigure}
  \caption{\label{fig:variable-whole}\textbf{The variable gadget}. In order to
    model 3-SAT we want the gadget to have two exact covers,
    corresponding to the variable being TRUE or FALSE, and for it to
    have several outputs in the same state as the variable, and
    several outputs in the opposite state. We construct the gadget out
    of smaller subgadgets. \textbf{(a)}, \textbf{(b)} and \textbf{(c)} show the RGB-duplicator
    subgadget, where all 3 links have the same state. \textbf{(d)}, \textbf{(e)} and
    \textbf{(f)} show these combined to form the duplicator subgadget which
    has several red links which must all be in the same state in any
    exact cover. Finally, \textbf{(g)} and \textbf{(h)} show the full variable
    gadget. Although we have shown four links nodes in $A$ and $B$, we
    can extend the duplicator subgadgets arbitrarily to ensure we
    have enough link nodes to enable connections to all the necessary clauses. Note the variable
    gadget in the FALSE state is just the mirror image of the variable
    gadget in the TRUE state shown in \textbf{(h)}.}
\end{figure*}
  \smallskip
  \textbf{The garbage collection gadget.}
  After the main part of our construction, we need to `clean
  up':  there will be some unused links from the
  variables, and there will be some parts of the construction we cannot
  fully `control' -- for example, some clauses may be `satisfied
  multiple times' (several of the literals are TRUE).
 
  The garbage collection gadget allows us to extend any exact cover of
  the variables, wires and clauses to an exact cover of the whole
  syndrome.  The gadget is shown in Figure~\ref{fig:xor-whole}.

  \smallskip
  \textbf{The crossing gadget.}  In our construction we will need to
  allow wires to cross or `pass' over each other. Whilst we have a
  construction which does this, we defer its construction to
  Appendix~\ref{s:crossing-wire}. In the more
  realistic noise models such as phenomenological noise or circuit
  level noise, it is not necessary -- in both of these models we have
  a third dimension (time) that we can use to allow wires to `hop'
  over each rather than needing to cross in the plane.

\subsection{An Example of the Construction}\label{s:example-construction}
Let us illustrate the construction with an example.  Suppose the 3-SAT
formula $F=(X_1\vee \bar{X}_2\vee X_4)\wedge (X_2\vee X_3\vee
X_5)\wedge (\bar{X}_2\vee\bar{X}_3\vee \bar{X}_4)$. We construct
the syndrome as follows (see Figure~\ref{fig:construction}) spacing the gadgets sufficiently far apart that we can join them as required by the formula.

  \begin{enumerate}
  \item
    First put in the variable gadgets;
  \item
    then the clause
    gadgets;
  \item
    then join the variables to the clauses with
    wires.
  \item
    Next add the garbage collection gadget;
  \item then
    join the unused variable outputs to the garbage collection gadget with wires.
  \item
    Finally, join the clause outputs to the garbage collection gadget with wires.
  \end{enumerate}

  Now that we have constructed the syndrome we need to show that it has an
  exact cover if and only if $F$ is satisfiable. 

  First suppose that there is an exact cover of this syndrome. Then each
  gadget must be covered exactly. In particular, each variable gadget
  must be in the TRUE or FALSE state, so each of its outputs must be
  in the corresponding state. Since each wire is covered exactly
  its ends are in the same state, which means that the inputs to the
  clauses are in the same state as their corresponding
  literal. Finally, since each clause gadget is covered exactly, it
  must have at least one TRUE input. In other words, this allocation
  of states to variables makes the formula $F$ TRUE; that is, the
  formula $F$ is satisfiable.

  To prove the converse suppose that the formula $F$ is satisfiable.
  Then there is an assignment of TRUE/FALSE to the variables that
  makes $F$ TRUE.  Put each variable gadget in the corresponding TRUE
  or FALSE state (i.e., choose the corresponding exact cover of each
  variable gadget).  Cover the wires between the variables and clauses
  exactly. Since we chose the state of the variables to satisfy the
  formula $F$, each clause gadget has at least one TRUE input so we
  can cover each of them exactly. Next cover the remaining wires
  exactly. Finally, cover the garbage collection gadget exactly (so far we have not
  explained why this can be done -- the details are given later). We
  have now covered the whole syndrome exactly, and in particular, we
  see that the syndrome does have an exact cover.

  This essentially completes the proof of Lemma~\ref{l:given_weight}
  and Theorem~\ref{t:min_weight} showing that decoding the colour code is NP-hard -- all that remains is to to show how
  to construct the various gadgets we described and to prove that they
  have the claimed properties. We do this in the next section.

  \section{Details of the gadgets}
  Before we go into the details of the individual gadgets we make the
  definition of the \emph{exact cover of a gadget} more
  precise. First, we define the \emph{error cluster of a gadget} to
  be the union of all error components meeting any defect of the
  gadget. This obviously contains all defects of the gadgets but may
  also contain some of the gadget's partner link nodes -- in fact it
  contains all the partner link nodes for the links that are
  TRUE. However, it cannot contain any other nodes from other
  gadgets. We say that a gadget is \emph{exactly covered} if the size
  of its error cluster equals the number of defects the error
  cluster meets.

  \subsection{The wire gadget}\label{ss:wire}
  A wire gadget is a sequence of red nodes $x_1,x_2,x_3,\dots,x_{2k}$ where
  consecutive nodes are at distance two, non-consecutive nodes are at
  distance at least three, the total number of nodes is even, and the
  only link nodes are $x_1$ and $x_{2k}$. See Figure~\ref{fig:wire-empty}
  for example. As mentioned above the key point is that, in any exact
  cover, the two link nodes of the wire need to be in the same state
  (both TRUE or both FALSE). This can be seen in two ways. One is to
  start at the link node at one end of the wire, setting it TRUE or
  FALSE and then observing that the errors in an exact cover are
  forced (see Figures~\ref{fig:wire-true}
  and~\ref{fig:wire-false}). Alternatively, since the wire's error
  cluster contains no blue or green defects, it must contain an even
  number of red defects, which means that either it contains both
  partner link nodes, or neither. In other words the two link nodes
  are either both TRUE, or both FALSE.

  Although we have insisted that wires have even length, this does not impede out ability to join the link nodes as required by our construction -- we discuss this in Section~\ref{s:combining}.

  Finally, we have chosen to insist that all the nodes in a wire are red. We
  have done this because it ensures that we can always join two link
  nodes in gadgets with a wire (recall link nodes are always red). Of
  course, we could have chosen any colour, and sometimes when
  constructing a gadget out of subgadgets we will join them by
  blue or green `wire-like' structures. 
  \subsection{The clause gadget}
 
  The clause gadget is shown in
  Figure~\ref{fig:clause-empty}. Figure~\ref{fig:clause-all-false}
  shows what is forced (if the cover were to be exact) if all three
  inputs are in the FALSE state, and we see that central node cannot
  be exactly covered; i.e., there is no exact cover in this
  case. Figures~\ref{fig:clause-1-true}, \ref{fig:clause-2-true}
  and~\ref{fig:clause-3-true} show an exact cover in the cases that
  one, two or three inputs are in the TRUE state.

  Note that we cannot control the state of the links to the garbage
  collection but we can see that the number of TRUE links from the
  gadget to garbage collection is exactly one less than the number of
  TRUE input links.
  \begin{figure*}
    \def\figbasename{xor-subgadget-false}
  \begin{subfigure}{0.3\textwidth}
  \includegraphics[width=\textwidth]{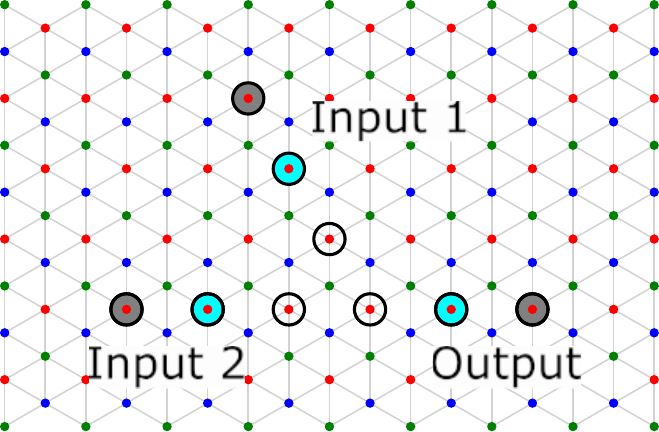}
      \caption{\centering\label{fig:xor-empty}Empty gadget}
  \end{subfigure}
    \hspace{0.03\textwidth}
  \begin{subfigure}{0.3\textwidth}
  \includegraphics[width=\textwidth]{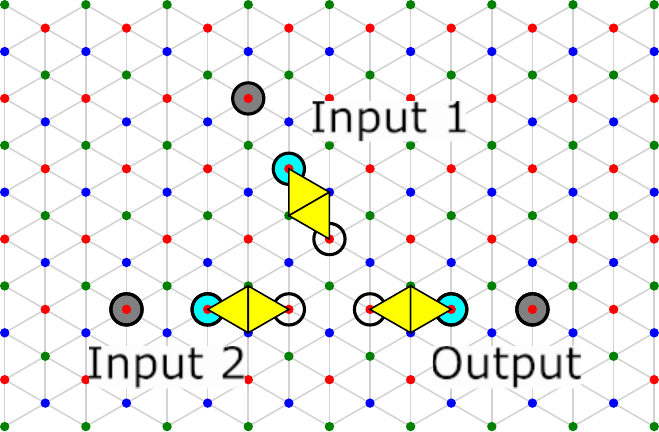}
      \caption{\centering\label{fig:xor-FF}Both inputs FALSE}
  \end{subfigure}
  \hspace{0.03\textwidth}
  \begin{subfigure}{0.3\textwidth}
  \includegraphics[width=\textwidth]{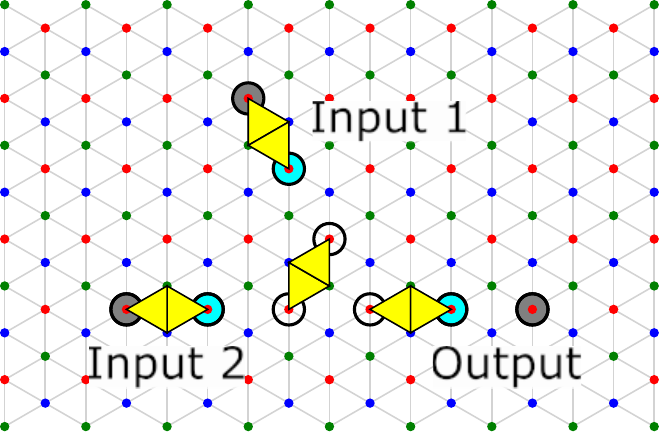}
      \caption{\centering\label{fig:xor-TT}Both inputs TRUE}
  \end{subfigure}

  \begin{subfigure}{0.3\textwidth}
  \includegraphics[width=\textwidth]{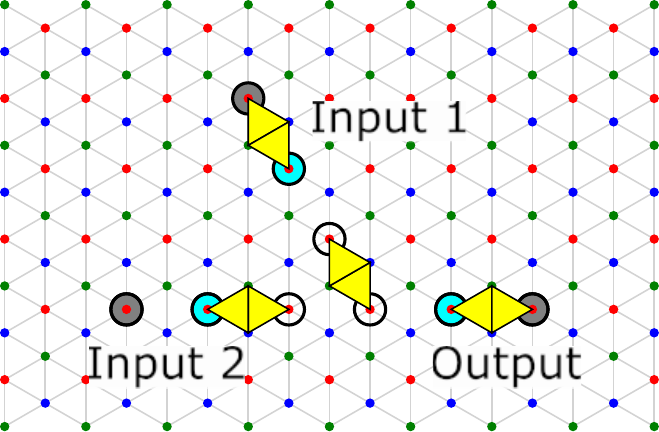}
      \caption{\centering\label{fig:xor-FT}Input 1 TRUE; Input 2 FALSE}
  \end{subfigure}
    \hspace{0.03\textwidth}
  \def\figbasename{xor-subgadget-true}
  \begin{subfigure}{0.3\textwidth}
  \includegraphics[width=\textwidth]{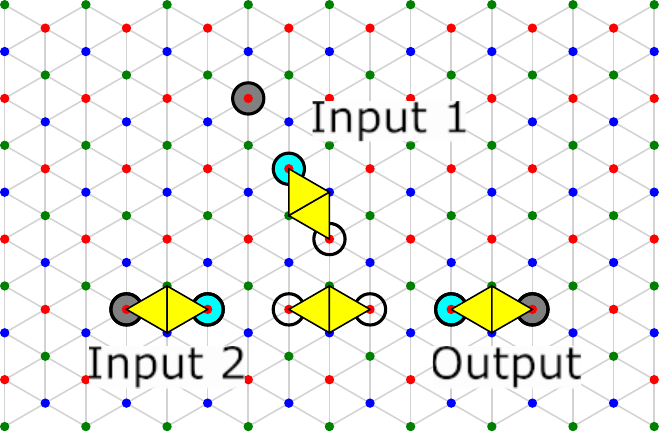}
      \caption{\centering\label{fig:xor-TF}Input 1 FALSE; Input 2 TRUE}
  \end{subfigure}
  \hspace{0.03\textwidth}
  \begin{subfigure}{\textwidth}
  \includegraphics[width=\textwidth]{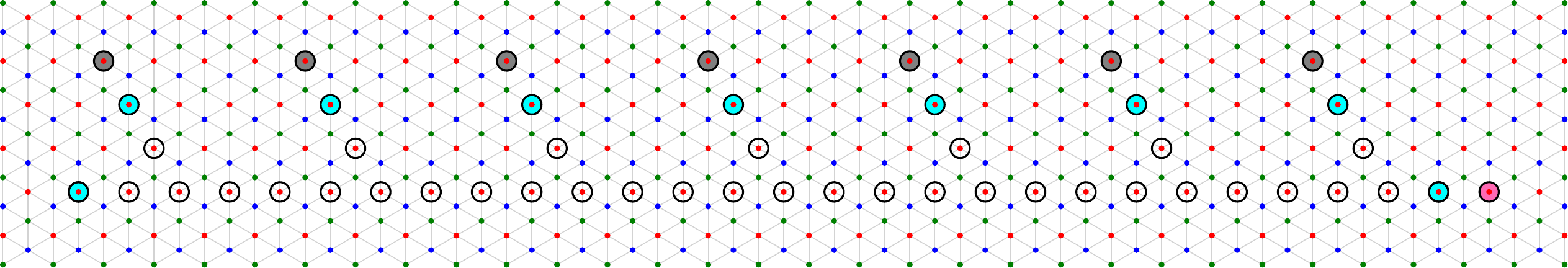}
      \caption{\centering\label{fig:full-gc-xor}The full garbage collection gadget}
   \end{subfigure}
  \caption{\label{fig:xor-whole}\textbf{The garbage collection gadget
      and its component XOR-subgadget}. \textbf{(a)} Shows the
    XOR-subgadget itself; \textbf{(b)-(e)} show that the (forced)
    exact covers showing that the output is indeed the XOR of the two
    inputs; \textbf{(f)} shows the full garbage collection gadget. When defining this gadget we ensure that the total number of red defects across our whole construction is even -- we do this by including or omitting the optional pink node as appropriate. This will depend on the formula $F$ and its associated syndrome. As
    there is no partner link node to the left of the gadget the input
    on the bottom left is always FALSE. The final output is the XOR of
    all the inputs, which will be TRUE, and thus use the pink node, exactly if there are an odd number of TRUE inputs to the garbage collection gadget. This case corresponds to an \emph{odd} number of red defects in the rest to the syndrome -- i.e., exactly the case that we included the pink node in our construction. 
     }
\end{figure*}

  \subsection{The variable gadget}\label{ss:variable-gadget}
  The variable gadget is more complicated then either of the previous
  gadgets. It is easiest to construct it out of smaller pieces. The
  first of these we call a RGB-duplicator subgadget and is shown in
  Figure~\ref{fig:rgb-subgadget-empty}. The key property of this
  subgadget is that all three of its link nodes have to be in the
  same state -- either all TRUE or all FALSE. Again this can be seen
  either by trying the two cases (Figures~\ref{fig:rgb-subgadget-true}
  and~\ref{fig:rgb-subgadget-false}), or observing that the
  subgadget's error cluster must have the same parity of red, green
  and blue defects, so either contains all three partner link nodes, or none
  of the partner link nodes. These two cases correspond to the link
  nodes all being TRUE or all being FALSE respectively.

  The RGB-duplicator subgadget duplicates a red link to a blue link
  and a green link. We actually want to be able to duplicate red links
  to red links, but this is now easy: we just join RGB-duplicators
  together. See Figure~\ref{fig:double-duplicator-empty} for an
  example -- all the red links have to be in the same state (see
  Figures~\ref{fig:double-duplicator-true}
  and~\ref{fig:double-duplicator-false}). We call this a duplicator
  subgadget. It is important to note that we can repeat this pattern arbitrarily many
  times to get as many red links as we want.
  
  Finally, to construct the variable gadget we join two duplicator
  gadgets together using a blue `wire-like' path containing an odd
  number of nodes -- see Figure~\ref{fig:variable-empty}. We see that
  all inputs in the set $A$ have the same state, and all inputs in set
  $B$ have the opposite state; see Figure~\ref{fig:variable-true}. In
  the figures we have shown the example where there are four link
  nodes in $A$ and four in $B$ but, of course, since we can extend the
  duplicator subgadget arbitrarily we can make the sets $A$ and $B$
  as large as we like. In particular, we can ensure that we have
  enough link nodes of each type to join to all the occurrences of the
  variable in the clauses.

  Note, we also joined the two spare `blue links' by an odd length blue
  `wire-like' path as we have insisted that all links from a gadget
  are red (as this makes joining everything together easier).

  \subsection{The garbage collection gadget}

  We construct the garbage collection gadget as a chain of very simple
  XOR subgadgets. The XOR subgadget is shown in
  Figure~\ref{fig:xor-empty} and Figures~\ref{fig:xor-FF}
  to~\ref{fig:xor-TT} show that in any exact cover of the subgadget
  the output link-node is indeed an XOR of the input link-nodes. (This
  can also be seen by a parity argument -- the subgadget's error
  cluster contains no green defects so must contain an even number
  of red defects.)

  The full garbage collection gadget is shown in
  Figure~\ref{fig:full-gc-xor}. Note, the optional pink node at the end of the
  bottom row; we choose to include or omit this defect from the syndrome so as to
  ensure that the total number of red defects across the whole of our constructed
  syndrome is even.

  To work out the parity of the number of TRUE inputs, we could argue
  directly and just count them, but it is easier to observe that the
  total number of green defects is even since they only occur in
  variable gadgets and crossing wire gadgets (see
  Section~\ref{s:crossing-wire} in the Appendix) and there are an even
  number in each such gadget. Therefore, any exact cover of the other
  gadgets must contain an even number of red defects.  Hence, to find
  the parity of the number of TRUE inputs, we just count the total
  number of red defects in the other gadgets: if that has even parity
  then the garbage-collection gadget has an even number of TRUE
  inputs, and if that is odd then the garbage-collection gadget has an
  odd number of TRUE inputs. The way we chose to include or exclude
  the pink node means that, in either case, we can complete the
  cover to an exact cover of the whole syndrome.

  \begin{figure*}
    \def\figbasename{logical-operator}
  \begin{subfigure}{0.3\textwidth}
  \includegraphics[width=\textwidth]{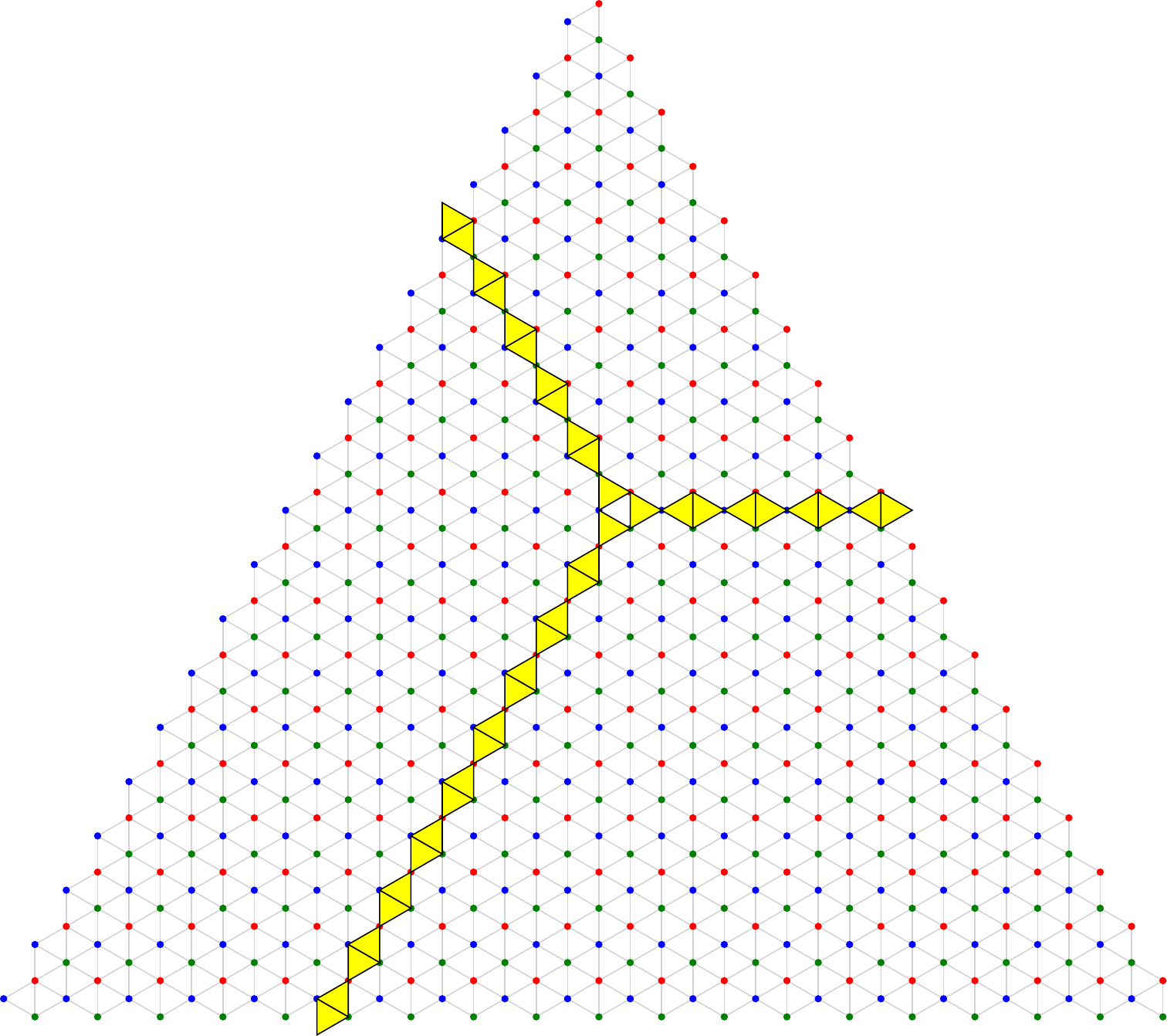}
      \caption{\centering\label{fig:logical-operator-complete}A minimum weight logical operator of weight $d$; in this case $d=39$}
  \end{subfigure}
  \hspace{0.03\textwidth}
  \begin{subfigure}{0.3\textwidth}
  \includegraphics[width=\textwidth]{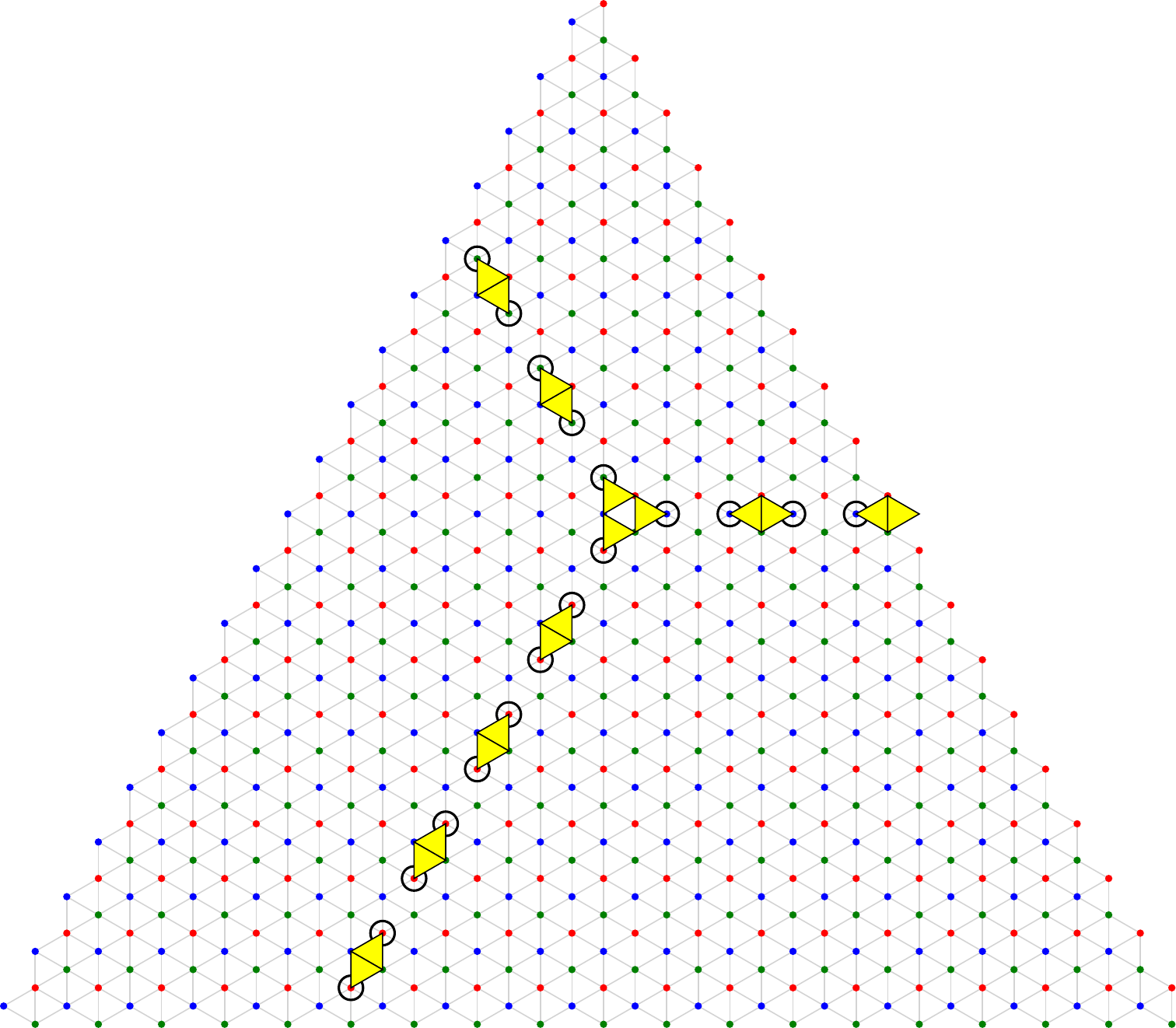}
      \caption{\centering\label{fig:logical-operator-true}Part of weight $\lfloor d/2\rfloor$; 19 in this example}
  \end{subfigure}
  \hspace{0.03\textwidth}
  \begin{subfigure}{0.3\textwidth}
  \includegraphics[width=\textwidth]{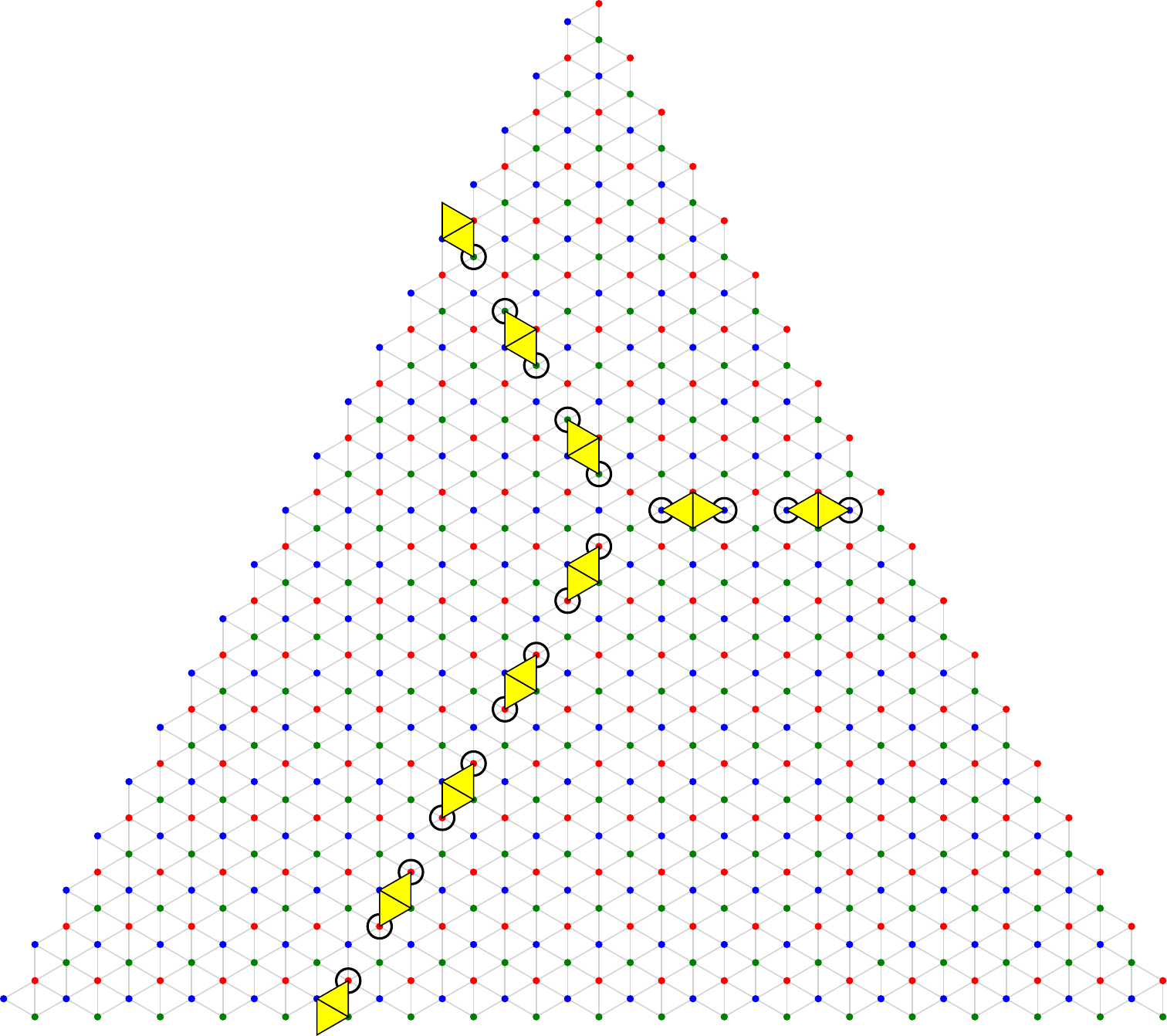}
      \caption{\centering\label{fig:logical-operator-false}Part of weight $\lceil d/2\rceil$; 20 in this example}
  \end{subfigure}

    \def\figbasename{logical-operator-with-links}
  \begin{subfigure}{0.3\textwidth}
  \includegraphics[width=\textwidth]{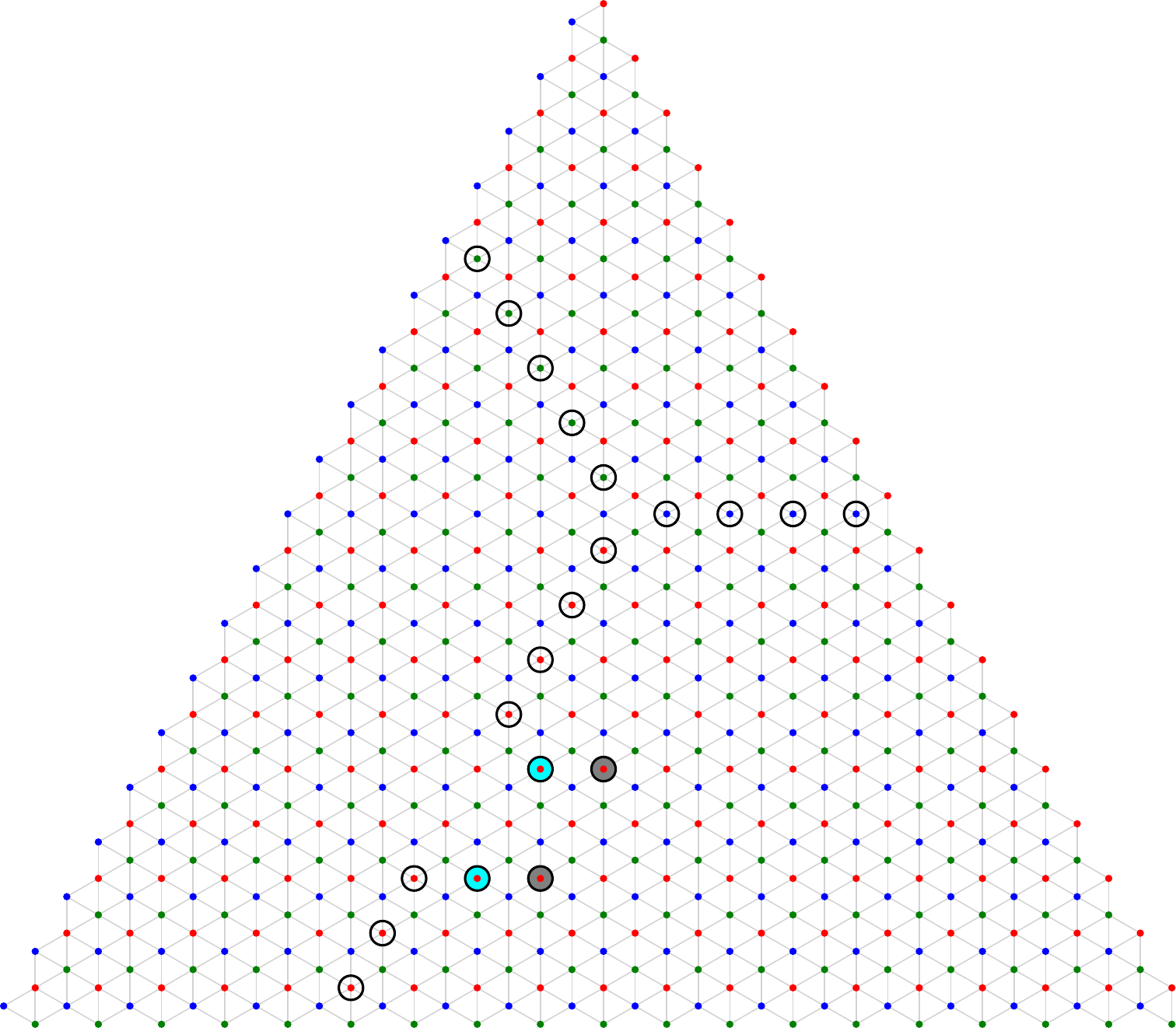}
      \caption{\centering\label{fig:logical-gadget-empty}Empty logical gadget}
  \end{subfigure}
  \hspace{0.03\textwidth}
  \begin{subfigure}{0.3\textwidth}
  \includegraphics[width=\textwidth]{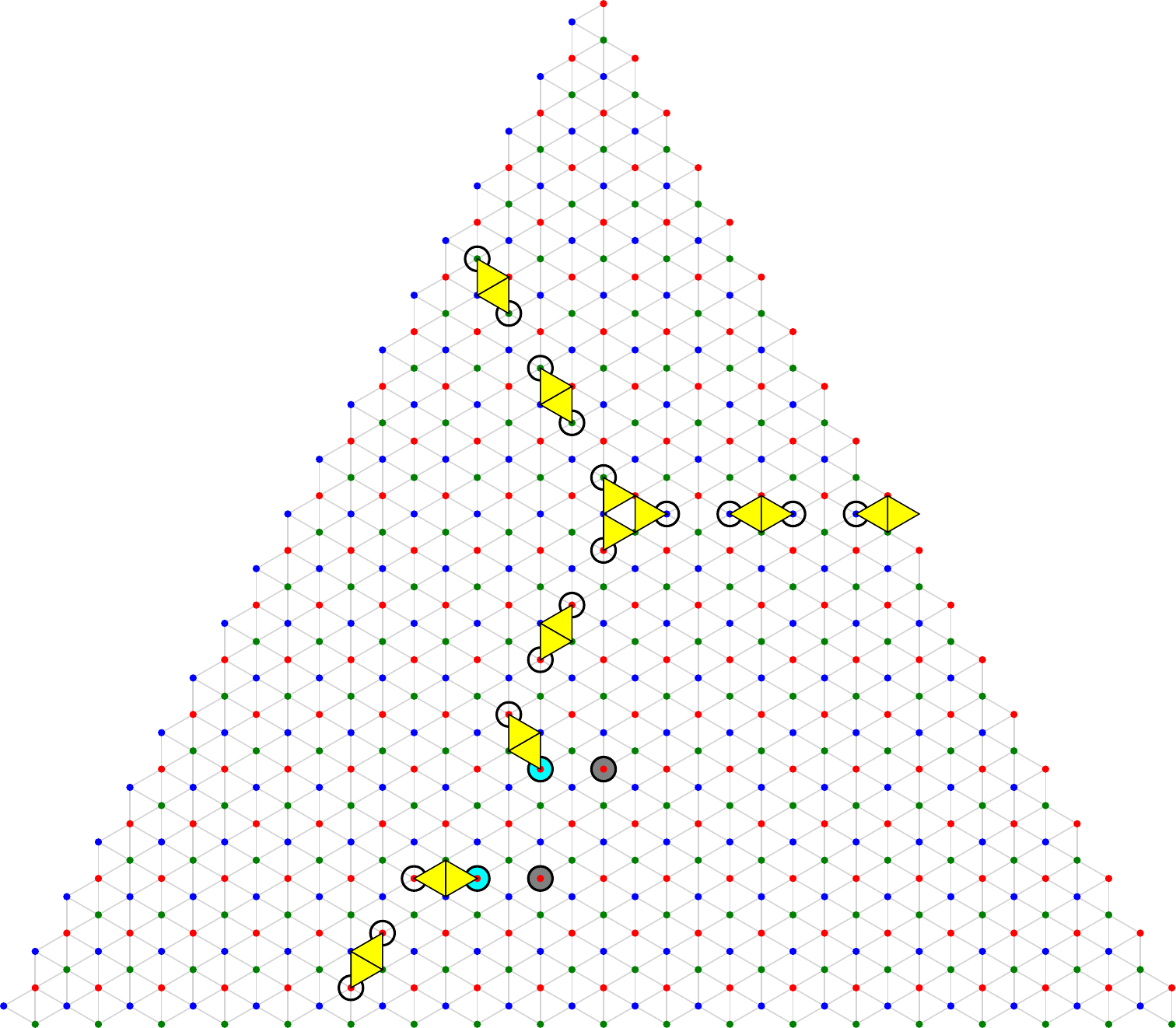}
      \caption{\centering\label{fig:logical-gadget-false}Link nodes FALSE}
  \end{subfigure}
  \hspace{0.03\textwidth}
  \begin{subfigure}{0.3\textwidth}
  \includegraphics[width=\textwidth]{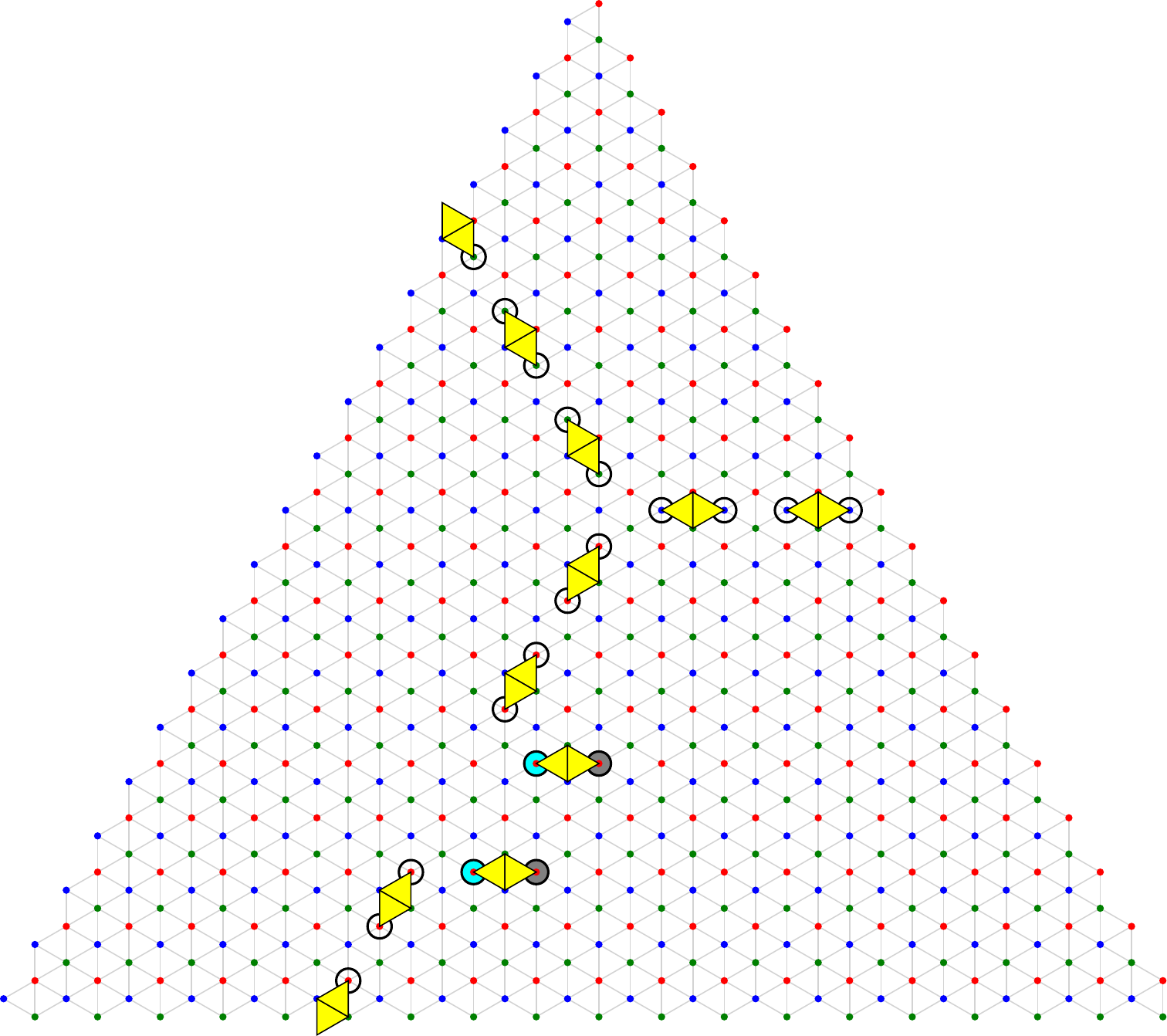}
      \caption{\centering\label{fig:logical-gadget-true}Link nodes TRUE}
  \end{subfigure}
  \caption{\textbf{A logical operator and corresponding logical
      gadget}. The aim of the logical gadget is to construct a
    syndrome with two covers of very similar weight -- one flipping
    the logical and one not -- such that deciding which of these two
    possibilities has lower weight relies on solving the minimum
    weight decoding problem given by Lemma~\ref{l:given_weight}. We
    take a weight $d$ logical operator \textbf{(a)} and split it into
    two parts \textbf{(b)} and \textbf{(c)} with similar weights and
    necessarily the same syndrome.  We take this syndrome and open it
    out as shown in \textbf{(d)}, joining the two link nodes to our
    previous construction with wires. \textbf{(e)} and \textbf{(f)}
    show the two covers we need to choose between -- they have opposite logical effects and differ in
    weight by one.  }
\end{figure*}

\subsection{The crossing-wire gadget}

The construction of this gadget is deferred to
Section~\ref{s:crossing-wire} in the Appendix as it is only needed for
the code capacity noise model -- in the more realistic noise models
such as phenomenological noise or circuit level noise we use the third
dimension (time) to allow wires to `pass' over each other.

\section{Combining the gadgets}\label{s:combining}
Now that we have all the gadgets we just need to combine them as
described in Section~\ref{s:example-construction} and shown in
Figure~\ref{fig:construction}. There are a few points to make. First,
we make the variable gadgets large enough that we have enough outputs
of each type ($X_i$ and $\bar X_i$). Then we place the clause gadgets
-- we need to place these sufficiently far from the variables, so that
there is room to route the wires from the variable gadgets to the
clause gadgets, and also from the variable and clause gadgets to
garbage collection.  There is one slightly subtle point: wires are
defined as being a path of an \emph{even} number of defects so we must
ensure that this constraint holds. However, it is easy to see that we
can deform an `odd length wire' slightly to make it an even length
wire so this is not an issue. We remark that in some other colour
codes it might not be possible to join two nodes by an even length
wire, but a very simple extra gadget solves this problem. We discuss
this in Section~\ref{s:other-colour-code-details} in the
Appendix.

Secondly, we consider the size of the error-set we construct. For the
NP-completeness result to be valid we need to show that the size of
our construction is polynomial in the size of the 3-SAT
formula. Suppose that the formula $F$ has $m$ clauses. Then the total
number of variable outputs is $O(m)$ and if we place the clauses at
distance $O(m^2)$ from the variables, and the garbage collection
gadget $O(m^2)$ beyond that, then it is straightforward to check that
we can place the necessary wires. In essence,
  we have to sort the $O(m)$ outputs from the variables into the order
  needed by the inputs to the clause gadgets, and we can do this
  sort using $O(m^2)$ swaps -- which corresponds to $O(m^2)$ times a wire hops over
  another (in phenomenological noise) or copies of the crossing gadget
  (in code capacity noise). This give a construction of size
$O(m^3)$.
We remark that whilst better
bounds are possible, they are not necessary for our results.

This completes the formal proof of Lemma~\ref{l:given_weight} and, thus, Theorem~\ref{t:min_weight}.

\section{The Logical}\label{s:logical}
In this section we prove Theorem~\ref{t:logical}. We construct a
simple gadget which we call the logical gadget.  We start with the
weight $d$ logical operator shown in
Figure~\ref{fig:logical-operator-complete}. We split this error set
into two parts of weight $\lfloor d/2\rfloor$ and $\lceil d/2\rceil$
which necessarily have the same syndrome (see
Figures~\ref{fig:logical-operator-true}
and~\ref{fig:logical-operator-false}).  We see that the syndrome is
essentially just an RGB-duplicator with some wires joining it to the
boundary.

We now cut the red wire and open it up into two link nodes (see
Figure~\ref{fig:logical-gadget-empty}). Since this gadget has an odd
number of red defects, an odd number of blue defects but an even
number of green defects, any cover of the gadget must meet the
boundary: either just the green boundary, or both the red and blue
boundaries. These two options are shown in
Figures~\ref{fig:logical-gadget-false}
and~\ref{fig:logical-gadget-true}. Since, any cover of the logical
gadget must meet the boundary, and no single error can meet both the boundary and
a defect, it makes sense to say that a cover of the logical gadget is
exact if its error cluster has size one more than the number of
defects (i.e., $|\cE|=|\cS|+1$).  With this definition, we see that the
cover in Figure~\ref{fig:logical-gadget-false} is an exact cover, but
the cover in Figure~\ref{fig:logical-gadget-true} has one more error than
that.  Thus the gadget can only be covered exactly if both
its link nodes are FALSE.  

In essence, we would like to join this logical gadget to our previous
construction. It turns out to be more convenient to join it to the
previous construction applied to a different 3-SAT formula.  Suppose
we are given $F$, a 3-SAT formula, in variables $X_1, X_2,\dots,
X_n$. We construct a 3-SAT formula $F'$ with several extra variables,
including a variable $Y$, with the property that $F'$ is always
satisfiable if $Y$ is TRUE, and $F'$ is satisfiable with $Y$ FALSE if
and only if $F$ is satisfiable. As the proof that we can do this is
routine algebra we defer it to Lemma~\ref{l:3-sat-y} in
Appendix~\ref{s:modified-3-sat}.

Now we use the aforementioned methods to construct a syndrome for the
formula $F'$ -- with one extra constraint insisting that the
variable $Y$ has two `spare' outputs that are joined straight to the
garbage collection gadget. By our construction, this syndrome has the
property that it has an exact cover if and only if $F'$ is
satisfiable. We join this syndrome to our logical gadget by removing
the two spare wires from $Y$ to the garbage collection, and joining
those two outputs from $Y$ to the two link nodes of the logical gadget
with wires. Note that, since we have removed two identical inputs from
the garbage collection gadget we have not changed whether it can be
exactly covered.

Suppose that the original 3-SAT formula $F$ is satisfiable. Then $F'$
is satisfiable with $Y$ FALSE, so there is an exact cover of the whole
syndrome including the logical gadget. This cover is clearly a
minimum weight error-set. On the other
hand, suppose that the original 3-SAT formula $F$ is not
satisfiable. Then $F'$ is not satisfiable with $Y$ FALSE, but is
satisfiable with $Y$ TRUE. In this case, there is no exact cover
(since an exact cover of the logical gadget needs its inputs to be
FALSE, and an exact cover of the rest of the syndrome needs the
outputs of $Y$ to be TRUE). Therefore, the minimum weight error set
must contain at least one more error than the exact cover. However, we
can construct an error-set of that weight (one more than an exact
cover) by covering the main part of the syndrome exactly with $Y$ TRUE
(this is where we use the fact that the formula $F'$ is satisfiable
when Y is TRUE), and then by covering the logical gadget with one
extra error. This is the minimum weight decoding in this case.

The numbers of errors in these two possibilities have opposite parity so  
have opposite effects on the the $Z$-logical -- one flips it and one does not.
Hence, any algorithm that can decide
whether the minimum weight decoding flips the logical can be used to
solve 3-SAT. This completes the proof of Theorem~\ref{t:logical}.

\section{Discussion and outlook}

\subsection{Other colour codes}\label{s:other-colour-codes}
Most of what we have done would work on any other planar colour code,
although the proofs will be mildly more technical since there is less
symmetry. We briefly outline what would need to be done for these
other codes. The clause, wire, and garbage collection gadgets all
consist entirely of defects which all have the same colour (red), so
intuitively these should transfer to other planar colour
codes. However, there are some subtleties. First, the precise
definition of the wire gadget needs a very slight modification see
Appendix~\ref{s:other-colour-code-details}). More importantly, though,
it may not be possible to join two arbitrary red nodes by a wire
containing an even number of nodes (in graph theoretic terms, it can
be the case that a related graph known as the shrunk lattice is
bipartite). This is a potential issue in terms of being able to join
all the gadgets together, but it also can be an issue inside our
gadgets or subgadgets. For example, the XOR-subgadget (used in the
Garbage Collection gadget) contains an `odd length red cycle' and this
is not possible if the shrunk lattice is bipartite. However, it is
easy to solve this issue by introducing a (very simple) NOT-subgadget
implementing a logical NOT operation -- see
Appendix~\ref{s:other-colour-code-details}.

The variable gadget contains defects of all three colours, so is
slightly more complicated.  However, the only thing we need to check is
that we can construct an RGB-duplicator subgadget
(Figure~\ref{fig:rgb-subgadget-empty}) -- once we have this we can
easily construct a variable gadget. This appears to be possible in
general (although we may not be able to insist on a separated
syndrome, so the proofs may become a little more technical). For 
example for the $4.8.8$ colour code see
Figure~\ref{fig:4-8-8-RGB-duplicator} and
Appendix~\ref{s:other-colour-code-details}. An RGB-duplicator
subgadget is all that is needed to make the logical gadget so that
can also be constructed in more general colour codes.

The crossing wire gadget is our most complicated gadget (see Appendix~\ref{s:crossing-wire}), but recall it
is only needed for the code capacity model (wires passing over each
other is possible in the phenomenological noise model in any planar
colour code). We would expect it to be possible in general, but have
only verified it for the $4.8.8$ colour code.

An interesting direction is whether our approach could be extended to
colour codes in three or more dimensions.  In such codes, one Pauli
error type must be associated with syndromes that are fundamentally
not string-like, but rather membrane-like (or even higher dimensional
objects)~\cite{bombinGaugeColorCodes2015}.  As such, for the membrane-like Pauli error, we would not be
able to straightforwardly use our wire gadgets, nor many of the other gadgets which implicitly contain wires.  However, we only need
one Pauli error type to be string-like for an extension of our result
to be plausible.  In the three-dimensional colour code, we always have one Pauli
error type associated with string-like syndromes.  On the other hand,
there exist four-dimensional colour codes where all Pauli errors generate
membrane-like syndromes, and in this setting, there are more
significant obstacles to our proof.  Curiously, membrane-like
syndromes are known to be decodable via so-called single-shot decoders~\cite{bombinSingleShotFaultTolerantQuantum2015, campbellTheorySingleshotError2019, guSingleShotDecodingGood2024}, which prompts the question as to whether single-shot colour
code decoding problems are fundamentally computationally easier than
their two-dimensional cousins.  Beyond these topological observations,
we have not attempted to extend our results to any of these
higher-dimensional cases, but hope these remarks give the reader
signposts towards fertile research directions.

\subsection{Approximating the minimum weight decoding}\label{s:approx}

 For some NP-hard problems it is possible to prove that (assuming
 $\text{P}\not=\text{NP}$) not only is it not possible to find the
 `best' solution quickly, but that it is not even possible to find an
 approximate solution quickly. For example, it is easy to see that
 finding the assignment of variables to a 3-SAT formula that maximises
 the number of true clauses is NP-hard (as it immediately solves the
 3-SAT problem), but in fact, even finding a solution that is within
 $7/8$ of the maximum number of true clauses has been shown by H\r{a}stad~\cite{hastadOptimalInapproximabilityResults2001} to be NP-hard.
 However,
 for other NP-hard problems such as the Euclidean Travelling Salesman
 Problem it is possible to approximate arbitrarily closely in
 polynomial time -- see
 Arora~\cite{aroraPolynomialTimeApproximation1998} and
 Mitchell~\cite{mitchellGuillotineSubdivisionsApproximate1999}.

Although our construction uses 3-SAT, the inapproximability of 3-SAT
does not imply the corresponding result for decoding. The key point is
that when constructing a syndrome from a 3-SAT formula with $m$
clauses our construction uses a superlinear number of errors (we use
$\Omega(m^3)$). This means that an error proportional to $m$ is only
of size $o(|\cE|)$. Thus, the results in this paper do not show which
of these classes colour-code decoding belongs to.  We leave this as an
open question:
\begin{question}\label{q:approx}
  Given $\eps>0$, is it possible to find a polynomial time algorithm
  that, given a syndrome, finds an error set generating that syndrome of
  weight at most $(1+\eps)$ times the minimum weight decoding?
\end{question}
We remark that, in the question above, the run-time of any such
polynomial time algorithm will behave `badly' with respect to the
closeness of approximation $\eps$, but since the details get rather
technical we defer the discussion to Section~\ref{a:approx} in the Appendix. 

The next open question is about the effect on the logical. In
Section~\ref{s:logical} we constructed a syndrome for which it is
NP-hard to decide whether the minimum weight decoding flips the logical. However, the
minimum weight decoding contained more than $d/2$ errors. This is
above the theoretical decoding threshold, so there are error-sets of
this size where the minimum weight decoding will give the `wrong'
answer. Although this may feel like a weakening of our result, we
would like to emphasise both that good decoders will correctly decode `many'
error-sets with weight above the theoretical threshold, and that, for the surface
code, there does exist a polynomial time algorithm finding the minimum
weight decoding regardless of the error size.

\begin{question}\label{q:logical}
  Given a syndrome for the colour code with a minimum weight decoding
  of weight less than $d/2$ is there a polynomial time algorithm that can
  decide whether the minimum weight decoding flips the logical?
\end{question}

We remark that an algorithm giving a positive answer to
Question~\ref{q:approx} could be used to decode all errors with a
minimum weight decoding less than $(1-\eps)d/2$, but would still leave
this question open.

\subsection{Topological perspective and consequences}\label{s:colour-code-harder}

Our result that minimum weight decoding of the colour code is NP-hard is perhaps surprising given the deep structure and symmetries of the code that make it appear qualitatively different to general hypergraph matching. Moreover, since colour codes and surface codes are topologically very similar and closely related~\cite{kubicaUnfoldingColorCode2015}, this raises the question of what property makes them so distinct with regards to decoding complexity. 

In the surface code $X$ and $Z$ errors both generate string-like defect patterns.  Such string-like patterns also occur in the colour code, most explicitly they are seen in our construction of wire gadgets. More generally, gadgets which have all their defects of the same colour (such as our wire, clause, and garbage collection gadgets) can also be decoded efficiently with a matching decoder. Indeed, if our gadgets \textit{only} used string-like defect patterns, this would be highly suggestive that efficient and exact minimum weight decoding would be possible by finding minimum weight matchings in graphs.

Crucially, however, our variable gadget makes use of the RGB-duplicator (Figure~\ref{fig:rgb-subgadget-empty}) which can be thought of as red, green and blue strings meeting and annihilating. It is this interaction between strings of all three colours that cannot be perfectly captured in polynomial time. In fact, this is the crux in the construction of our syndrome that single-handedly transforms the complexity of the minimum weight colour code decoding problem globally.  As such, the topological perspective on our proof is that it rests on the two-dimensional colour code possessing string-like defect patterns that can interact (see Figure~\ref{fig:rgb-subgadget-empty} for example). Other structural information, for example the colour information that accompanies defects is, from a complexity perspective, insufficient to overcome this.

In proving that minimum weight decoding of the colour code is NP-hard, we can say definitively that any transition from surface codes to colour codes will involve a greater degree of approximation in decoding than might otherwise have been anticipated. 
Moreover, since we have shown that fast, exact colour code decoders do not exist, research should focus on improved heuristic and approximate methods balancing accuracy with the speed and scalability necessary in the large-scale, real-time setting. 

\section{Acknowledgements}

  We would like to thank Earl Campbell for discussions and comments on
the manuscript and, in particular, insights into higher dimensional
colour codes and single-shot decoding.
We would also like to thank Ben Barber for discussions on the underlying combinatorics, and 
Joan Camps for further comments on the manuscript.
Finally, we thank Steve Brierley and our colleagues at Riverlane for creating a stimulating environment for research.
\vfill

\pagebreak
\clearpage
   
\appendix

\begin{figure*}
    \def\figbasename{crossing-subgadget-green}
  \begin{subfigure}{0.3\textwidth}
  \includegraphics[width=\textwidth]{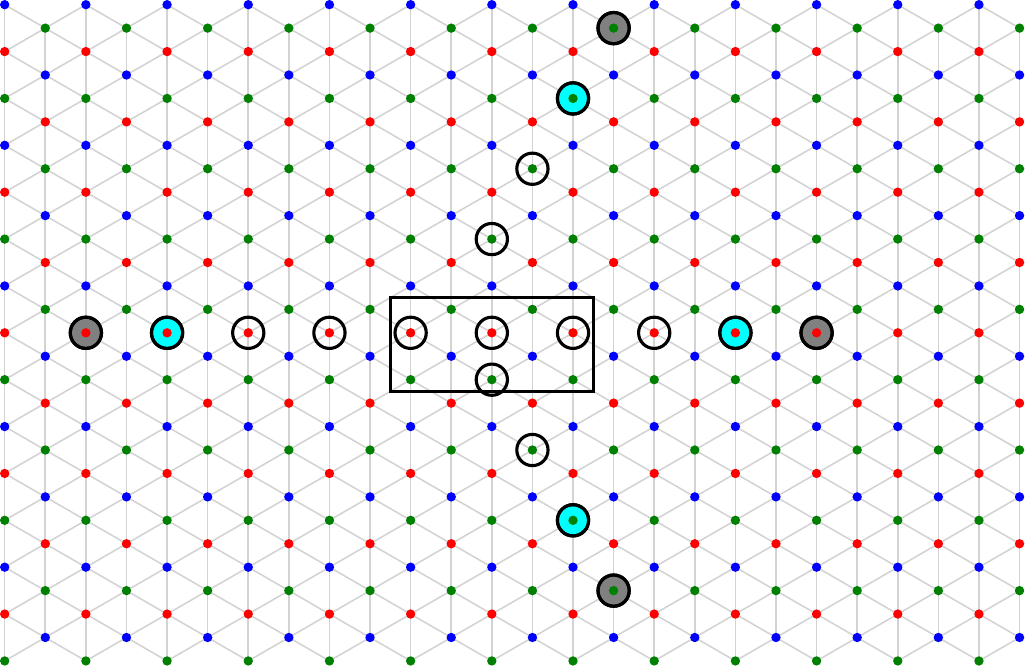}
      \caption{\label{fig:crossing-subgadget-empty}Empty multi-colour crossing subgadget}
  \end{subfigure}
  \hspace{0.03\textwidth}
  \begin{subfigure}{0.3\textwidth}
  \includegraphics[width=\textwidth]{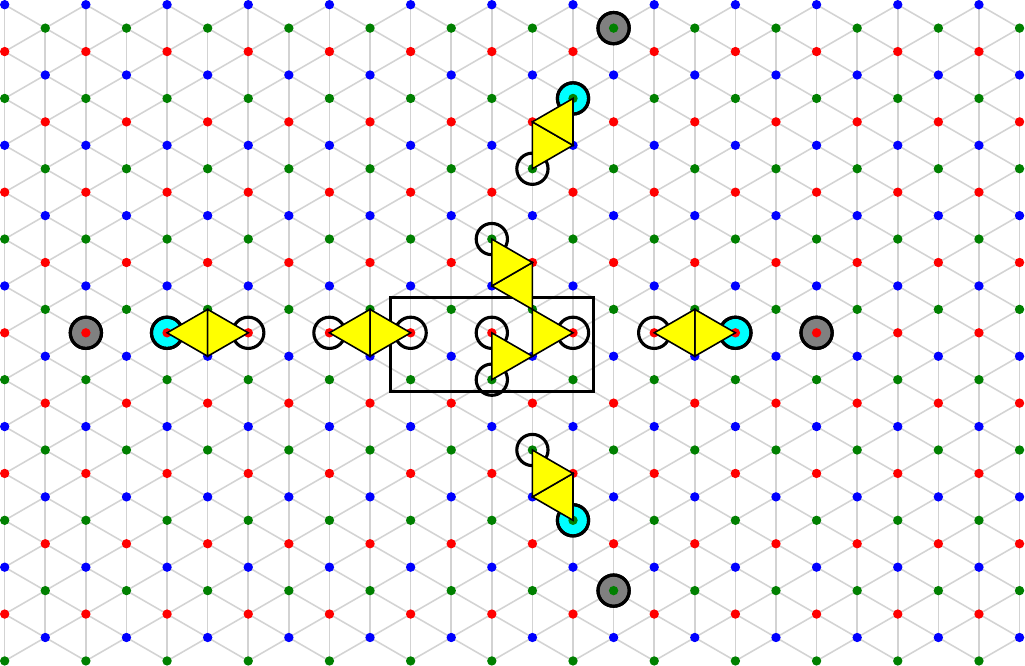}
      \caption{\label{fig:crossing-subgadget-ff}Red links and green links both FALSE}
  \end{subfigure}
  \hspace{0.03\textwidth}
  \begin{subfigure}{0.3\textwidth}
  \includegraphics[width=\textwidth]{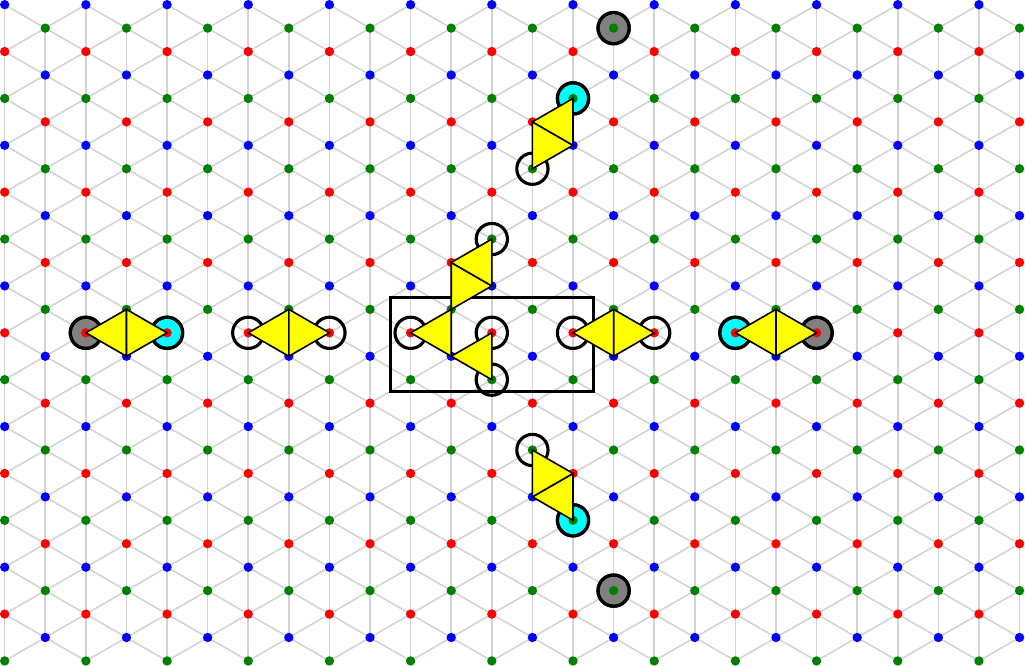}
      \caption{Red links TRUE and green links FALSE}
  \end{subfigure}
  \begin{subfigure}{0.3\textwidth}
  \includegraphics[width=\textwidth]{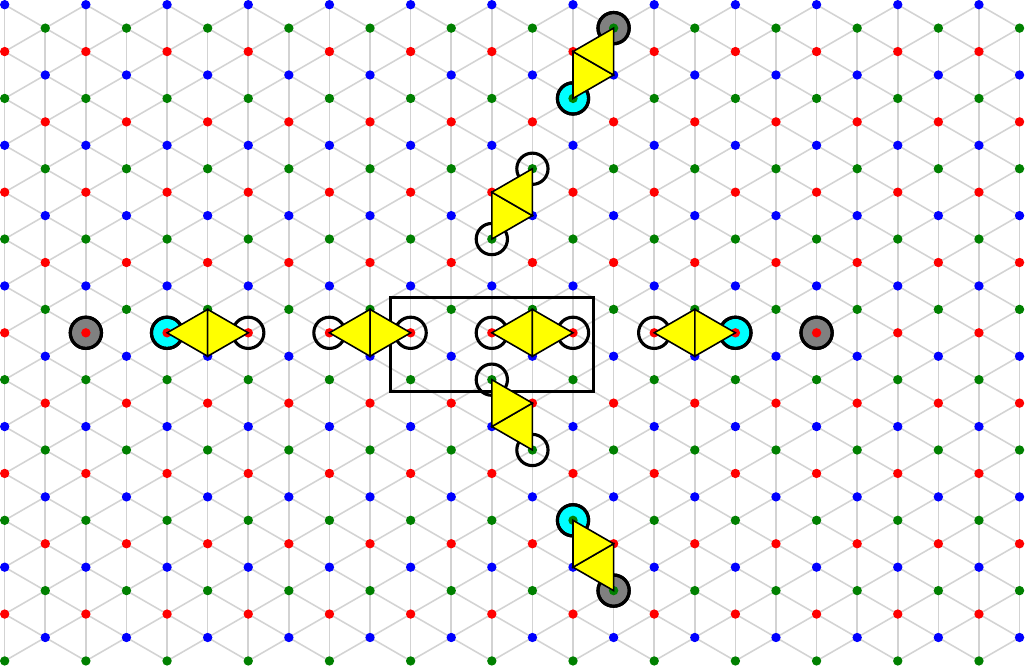}
      \caption{Red links FALSE and green links TRUE}
  \end{subfigure}
  \hspace{0.03\textwidth}
  \begin{subfigure}{0.3\textwidth}
  \includegraphics[width=\textwidth]{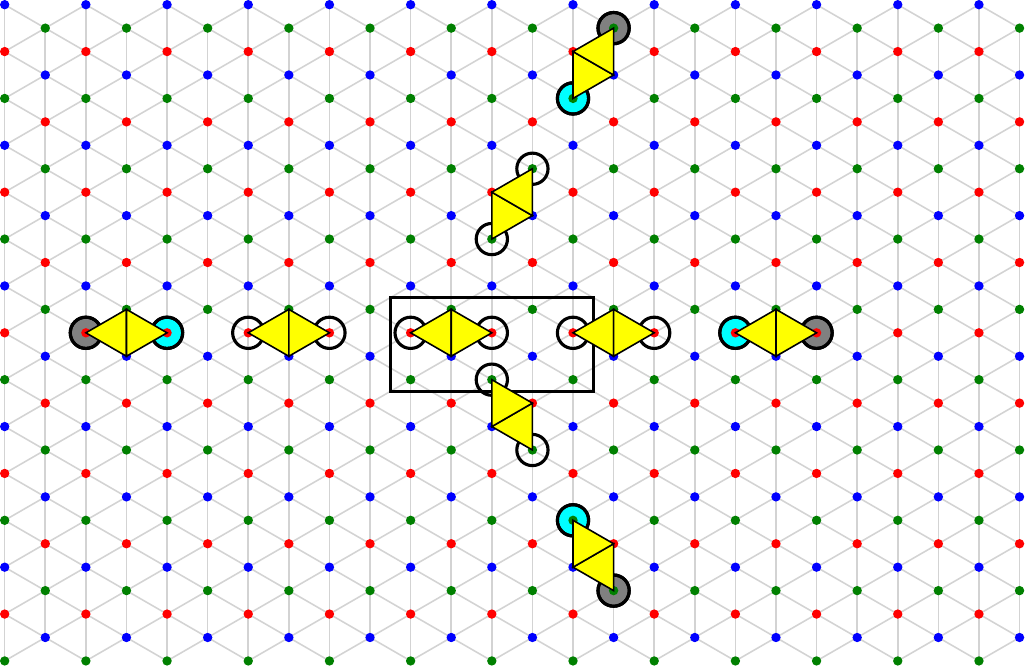}
      \caption{\label{fig:crossing-subgadget-tt}Red links and green links both TRUE}
  \end{subfigure}
    \def\figbasename{full-crossing-gadget}
  \begin{subfigure}{0.45\textwidth}
  \includegraphics[width=\textwidth]{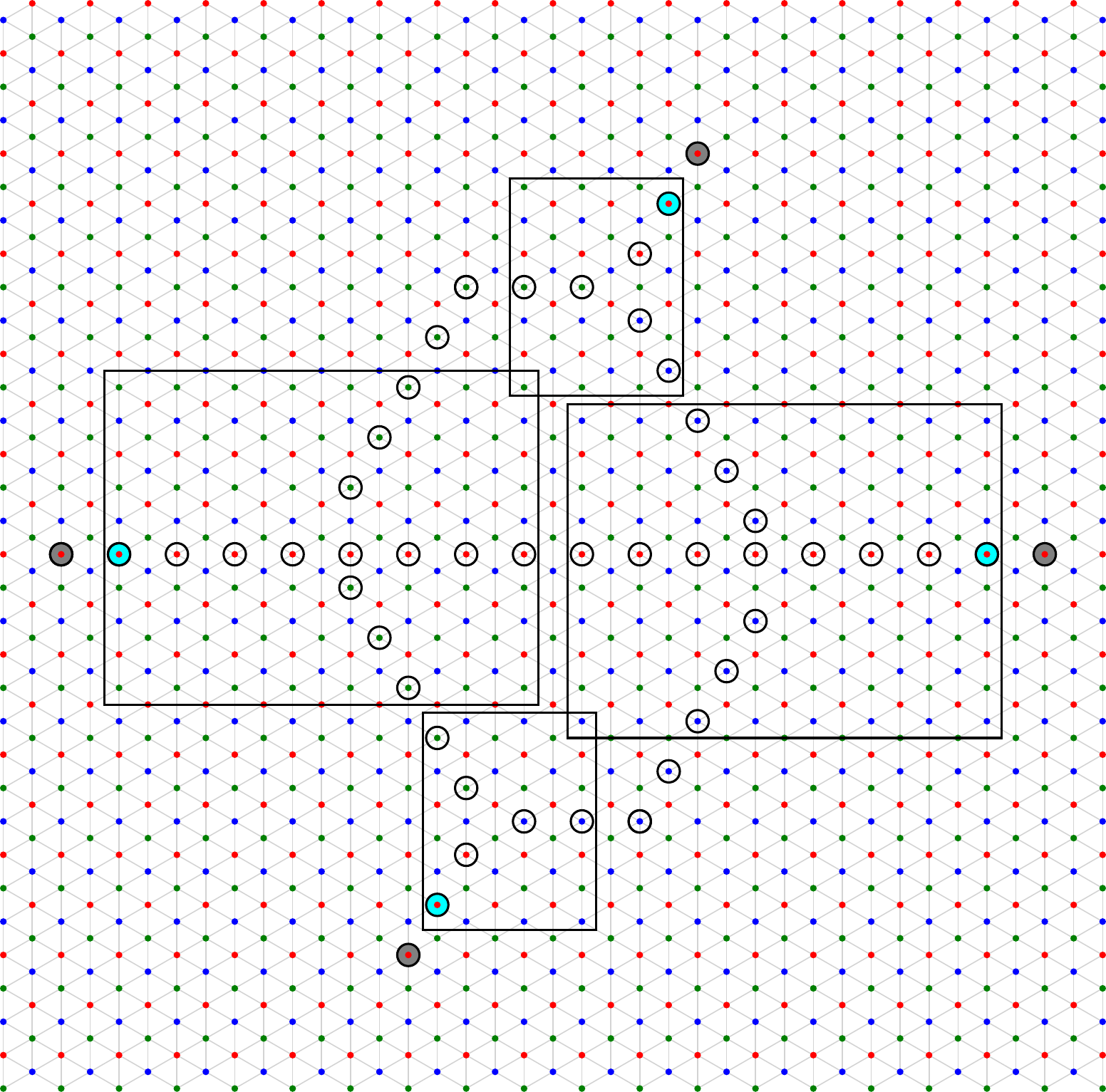}
      \caption{\label{fig:full-crossing-empty}\centerline{Empty full crossing gadget}\newline}
  \end{subfigure}
  \hspace{0.03\textwidth}
  \begin{subfigure}{0.45\textwidth}
  \includegraphics[width=\textwidth]{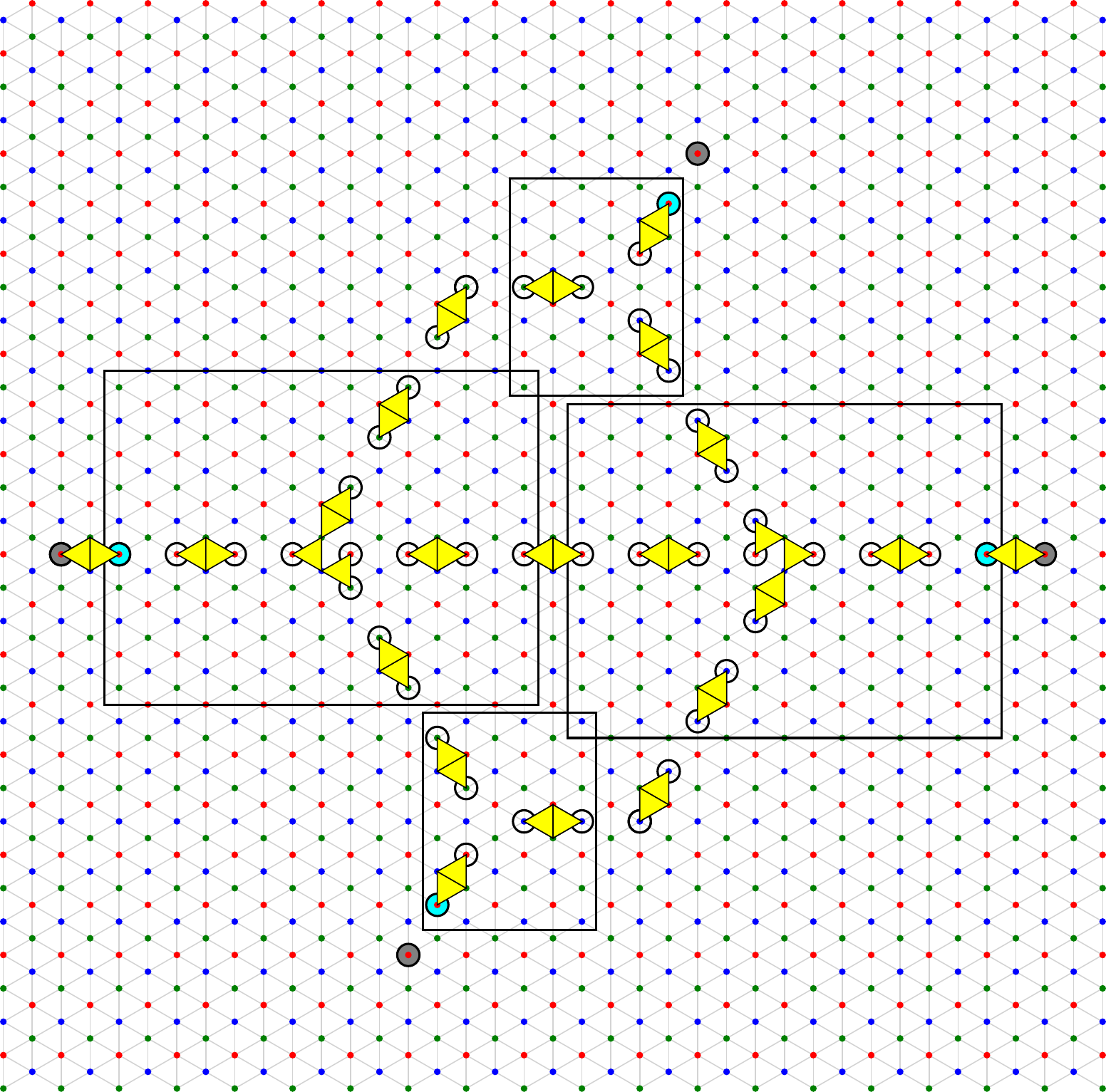}
      \caption{\label{fig:full-crossing-example}An example with the red links TRUE and the green links FALSE.}
  \end{subfigure}
  \caption{\label{fig:full-crossing}\textbf{The full crossing gadget.} In order to maintain the correctness of the wires through a crossing in the plane, we want the
    link nodes at the top and bottom to have the same state, and the
    link nodes at the left and right to have the same state. We
    construct this out of four subgadgets. \textbf{(a)}-\textbf{(e)} show the
    multi-colour crossing subgadget, which has the properties we want
    except the two `wires' have \emph{different} colours. \textbf{(f)} Shows
    the full crossing gadget (where both wires are red) which is
    formed by combining two multi-colour crossing subgadgets with two
    RGB-duplicators (the boxes show the four subgadgets). \textbf{(g)} shows
    one of the possible crossing states (the other three possibilities
    are omitted as they are simple combinations of the states of the
    four subgadgets).}
\end{figure*}

  \begin{figure*}[t]
  \begin{subfigure}{0.3\textwidth}
  \includegraphics[width=\columnwidth]{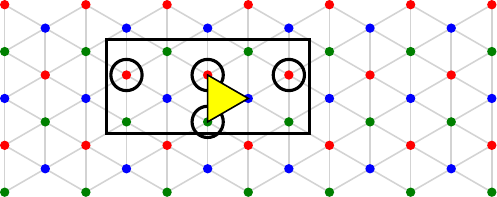}
      \caption{\label{fig:syndrome-box-one-error}One error meets two
        defects, but introduces a blue defect}
  \end{subfigure}
  \hspace{0.05\columnwidth}
  \begin{subfigure}{0.3\textwidth}
  \includegraphics[width=\columnwidth]{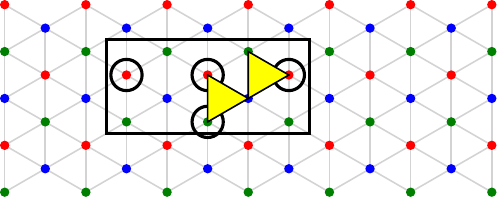}
      \caption{\label{fig:syndrome-box-two-errors-a}One way of cancelling the blue node\newline}
  \end{subfigure}
  \hspace{0.05\columnwidth}
  \begin{subfigure}{0.3\textwidth}
  \includegraphics[width=\columnwidth]{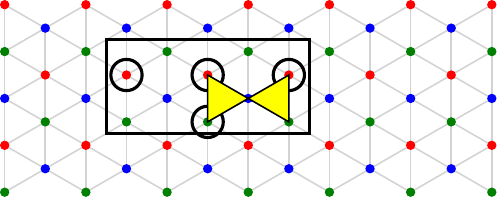}
      \caption{\label{fig:syndrome-box-two-errors-b}Another way of cancelling the blue node}
  \end{subfigure}
  \caption{\label{fig:syndrome-box}\textbf{The syndrome box from the
      multi-colour crossing subgadget.} The defects in the box are not
    separated, but we prove (Lemma~\ref{l:syndrome-region}) that, in
    any cover, at least four errors meet this box so the notion of an
    exact cover is preserved. If there is a single error that meets
    two defects from this box then, without loss of generality, we
    have \textbf{(a)}. This introduces a blue defect that needs to be
    cancelled. There are two ways of doing this while also cancelling the red defect on the right which are shown in \textbf{(b)}
    and \textbf{(c)}. Then we have a green defect on the right that
    needs to be cancelled, and the red defect on the left still needs
    to be covered. In total at least four errors must meet a node of
    the syndrome box.}
  \end{figure*}

  \section{The crossing-wire gadget}\label{s:crossing-wire} 
  This gadget is more complicated than our previous gadgets, but much of it is built from
  subgadgets we have seen before.
  
  We start with subgadget shown in
  Figure~\ref{fig:crossing-subgadget-empty} which we call a
  multi-colour crossing subgadget. We see that it breaks our rule
  that the syndrome contains no defects at distance one -- it contains
  one pair of defects at distance one. Thus, \emph{a priori}, it might
  be possible for an error set generating this syndrome to have size
  (one) less than the size of the syndrome. However, we claim that at
  least four errors must meet a node in the central box shown in
  Figure~\ref{fig:crossing-subgadget-empty}. We call the box the
  \emph{syndrome box}. This is just a short case analysis which we
  give in the following lemma.

  \begin{lemma}\label{l:syndrome-region}
  In any cover (exact or otherwise) of a multi-colour crossing
  subgadget at least four errors meet a node of the syndrome box.
\end{lemma}
\begin{proof}
The first case is that no error meets more than one defect. Then as
before we have (at least) four errors meeting the defects in the box.

  The other case is that an error does meet more than one defect.
  Without loss of generality this looks like
  Figure~\ref{fig:syndrome-box-one-error}.  This introduces a blue
  defect which needs to be cancelled. If this cancellation does not
  also meet one of the other two red defects then they still need two
  errors, meaning we have at least four errors in total. So suppose,
  that the error cancelling this blue defect also meets a red
  defect. There are two ways this can occur shown in
  Figures~\ref{fig:syndrome-box-two-errors-a}
  and~\ref{fig:syndrome-box-two-errors-b}. Both of these introduce a
  green defect, so we need at least two more errors: one for the red defect on
  the left and one for the green defect we have introduced on the
  right.
\end{proof}

  Since all the defects in the multi-colour crossing subgadget which
  are outside the syndrome box are distance at least two from any node
  in the box, no error meeting a node in the box can also meet a
  defect outside the box. Thus the number of defects covered by the
  subgadget's error cluster is at most the number of errors in that
  cluster. In other words we still have the notion of an exact
  cover.

  We remark that, whilst it is convenient to have green link nodes in
  the subgadget they are only present during this intermediate step
  in constructing the full gadget; the full gadget will only have red
  link nodes.

  We claim that, for the multi-colour crossing subgadget,
  \begin{enumerate}
  \item any exact cover either covers both green partner `link nodes', or
    neither; i.e., the green link nodes are either both TRUE, or
    both FALSE.
  \item Any exact cover either covers both red partner link nodes, or neither;
    i.e., the red link nodes are either both TRUE, or both
    FALSE.
    \item All four combinations of these are possible.
  \end{enumerate}

  The easiest way to see the first two of these is a parity
  argument. Consider the subgadget's error cluster. Since the
  subgadget contains no blue defects, its error cluster also
  contains no blue defects, so must contain an even number of green
  defects, and an even number of red defects. It is easy to see that
  this implies the first two points. The four
  Figures~\ref{fig:crossing-subgadget-ff}-\ref{fig:crossing-subgadget-tt}
  show that all four combinations of these cases can occur.
 
Since all our wires are red we actually want red wires to cross red
wires. We do this by splitting one of the red wires into a green wire
and a blue wire, crossing them over the other red wire, and then
recombining. The splitting/recombining is done using the
RGB-duplicator subgadget we saw earlier
(Section~\ref{ss:variable-gadget} and
Figure~\ref{fig:rgb-subgadget-empty}). The full gadget is shown in
Figure~\ref{fig:full-crossing-empty}. Since we know how each subgadget
behaves, it is easy to see that this has the properties we want: the
link nodes at the top and bottom of the figure must have the same
state, and the link nodes at the left and right must have the same
state, and that all four cases can occur. We show one example in
Figure~\ref{fig:full-crossing-example}.

This completes the construction of the crossing-wire gadget.

\section{Modifying the 3-SAT formula}\label{s:modified-3-sat}

In this section we describe how to construct the modified 3-SAT formula $F'$ used in the proof of Theorem~\ref{t:logical} in Section~\ref{s:logical}. 

\begin{lemma}\label{l:3-sat-y}
  Suppose $F$ is a 3-SAT formula in variables
  $X_1,X_2,\dots,X_n$ with $m$ clauses. Then there is a 3-SAT formula $F'$ in variables
  $Y, X_1,X_2,\dots X_n, Z_1, Z_2,\dots Z_m$ with $2m$ clauses, such that
  $F'$ is (identically) TRUE if $Y$ is TRUE, and $F'$ is satisfiable with $Y$ FALSE
  if and only if $F$ is satisfiable.
\end{lemma}
\begin{proof}
 Suppose $F$ has clauses $C_1\dots C_m$ (so $F=\bigwedge_i
 C_i$). First form a `4-SAT' formula $F''=\bigwedge_{i=1}^m C_i'$ in
 variables $Y, X_1,X_2,\dots X_n$ by replacing each clause $C_i$ by
 the clause $C_i'$ defined by $C_i'=C_i\vee Y$. It is immediate that
 if $Y$ is TRUE then $F''$ is identically TRUE (so definitely satisfiable), and if $Y$ is FALSE then $F'$ is
 satisfiable if and only if the original formula $F$ is satisfiable.
 
  To complete the proof we convert $F''$ to a 3-SAT formula in a
  standard way. Add variables $Z_1,\dots, Z_m$ and for each clause
  $C_i'=X_{i_1}\vee X_{i_2}\vee X_{i_3}\vee Y$ form two clauses
  $X_{i_1}\vee X_{i_2}\vee Z_i$ and $X_{i_3}\vee Y \vee \bar Z_i$.
  This new formula has the claimed properties. This follows because
  the pair of clauses above are jointly satisfiable (i.e., we can find
  $Z_i$ to make them both TRUE) if and only if the 4-variable clause
  $C_i'=X_{i_1}\vee X_{i_2}\vee X_{i_3}\vee Y$ is satisfiable.
\end{proof}

\begin{figure*}[t]
  \def\figbasename{4-8-8-not-gadget}
  \begin{subfigure}{0.3\textwidth}
  \includegraphics[width=\columnwidth]{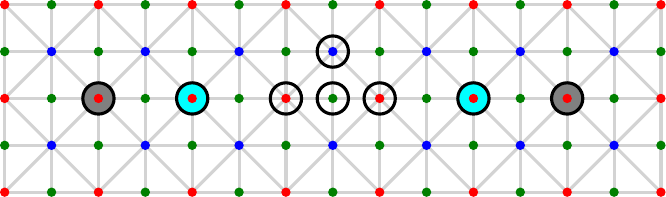}
      \caption{\label{fig:4-8-8-not-empty}Empty NOT-subgadget\newline}
  \end{subfigure}
  \hspace{0.03\columnwidth}
  \begin{subfigure}{0.3\textwidth}
  \includegraphics[width=\columnwidth]{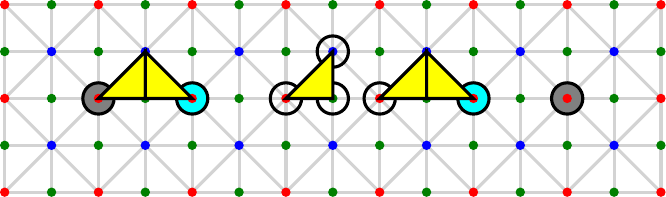}
      \caption{\label{fig:4-8-8-not-true}Input TRUE, output FALSE\newline }
  \end{subfigure}
  \hspace{0.03\columnwidth}
  \begin{subfigure}{0.3\textwidth}
  \includegraphics[width=\columnwidth]{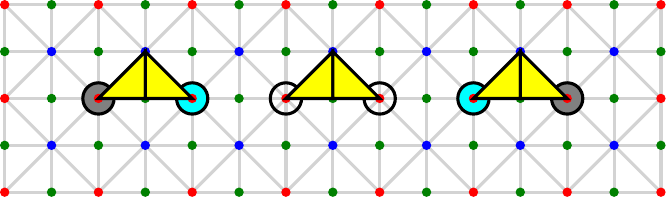}
      \caption{\label{fig:4-8-8-wire-true}A wire of the same length with both links TRUE}
  \end{subfigure}
  \caption{\label{fig:4-8-8-not-gadget} \textbf{The NOT-subgadget for
      the $\mathbf{4.8.8}$ colour code}. \textbf{(a)} the empty
    NOT-subgadget. Note the syndrome is not separated -- in this case
    we define a cover of this gadget to be exact if the number of red
    defects it covers is equal to the number of errors. \textbf{(b)}
    shows the exact cover when the input (bottom left) is TRUE. The
    case of the input being FALSE (not shown) is just the reversed
    gadget. \textbf{(c)} shows a wire of the same length and we see
    that that both ends have the same state (the case of both ends
    TRUE is shown) so the parity changes when replacing a bare wire by
    a NOT-gate.}
\end{figure*}

\section{Other noise models}\label{a:other-noise}
We have shown that finding a minimum weight decoding in the code-capacity
model with independent equally likely errors for the colour code is
NP-hard. However, in reality we are much more interested in the case
of phenomenological noise or circuit level noise. In general, decoding
is harder in these cases, so we would expect them to also be
NP-hard. For the phenomenological case this is true -- we just take
the exact syndrome constructed here in one layer (time-step) of the lattice. Since
the minimum weight decoding set must lie in the layer our proof
immediately applies.

For circuit level noise the situation is more
  complicated, not least because there are many possible circuits that
  could be used. Again we consider the exact syndrome constructed
  above in one layer -- the only way our proof could fail is if the
  existence of errors such as hook errors mean there are other error sets generating
  the given syndrome that have even lower weight (higher probability)
  than the exact covers we discuss. Since hook errors cause defects in
  multiple layers, it is unlikely that they would cancel in the way needed
  to give a lower weight error set, unless the circuit were
  specifically contrived (almost certainly reducing the code distance in the process). 

Similarly, if we consider noise models where errors do not all have
the same probability then the problem \emph{could} become easier. For
example, in the unrealistic scenario where some errors have zero
probability the decoding problem could reduce to a much simpler
problem such as MWPM. However, we would not expect this to be the case
for any realistic noise model. Indeed, we have chosen to concentrate
on the equal probability case to emphasise that it is not the
\emph{weights} that make the problem hard.

\begin{figure*}[t]
  \def\figbasename{4-8-8-xor-subgadget}
  \begin{subfigure}{0.3\textwidth}
  \includegraphics[width=\columnwidth]{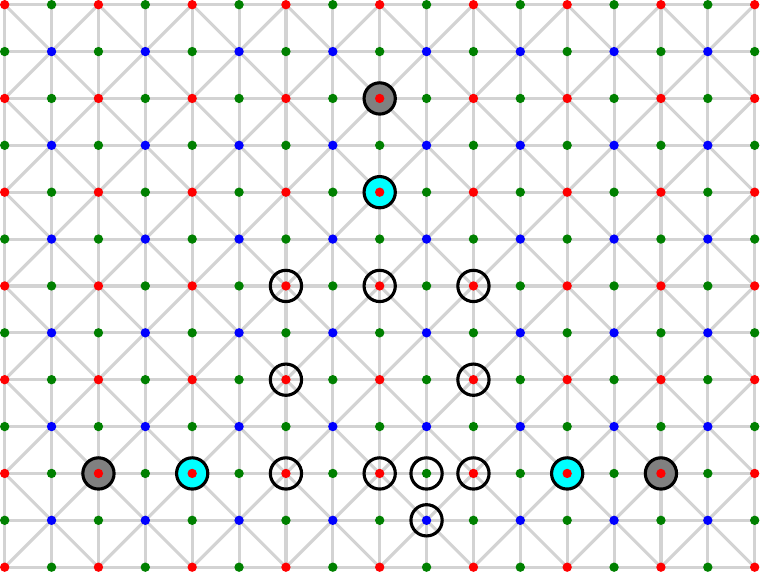}
      \caption{\label{fig:4-8-8-xor-empty}Empty XOR-subgadget}
  \end{subfigure}
  \hspace{0.03\columnwidth}
  \begin{subfigure}{0.3\textwidth}
  \includegraphics[width=\columnwidth]{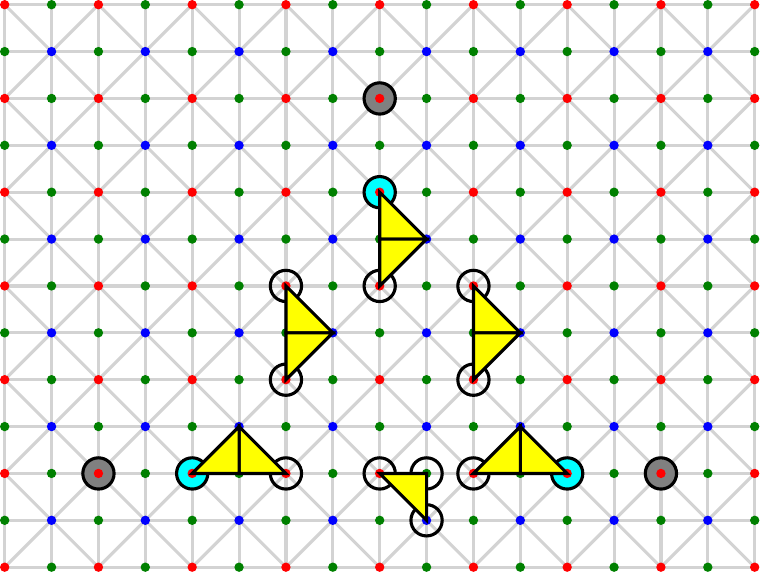}
      \caption{\label{fig:4-8-8-xor-FF}Both inputs FALSE}
  \end{subfigure}
  \hspace{0.03\columnwidth}
  \begin{subfigure}{0.3\textwidth}
  \includegraphics[width=\columnwidth]{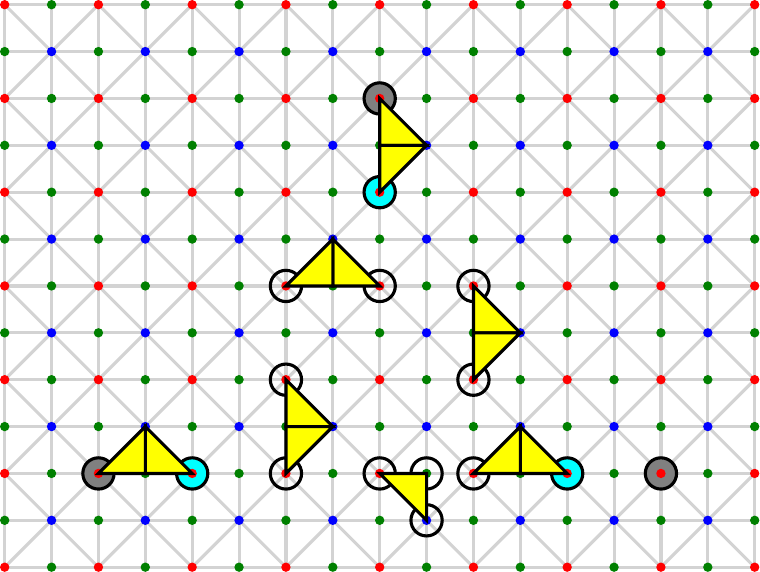}
      \caption{\label{fig:4-8-8-xor-TT}Both inputs TRUE.}
  \end{subfigure}
  \begin{subfigure}{0.3\textwidth}
  \includegraphics[width=\columnwidth]{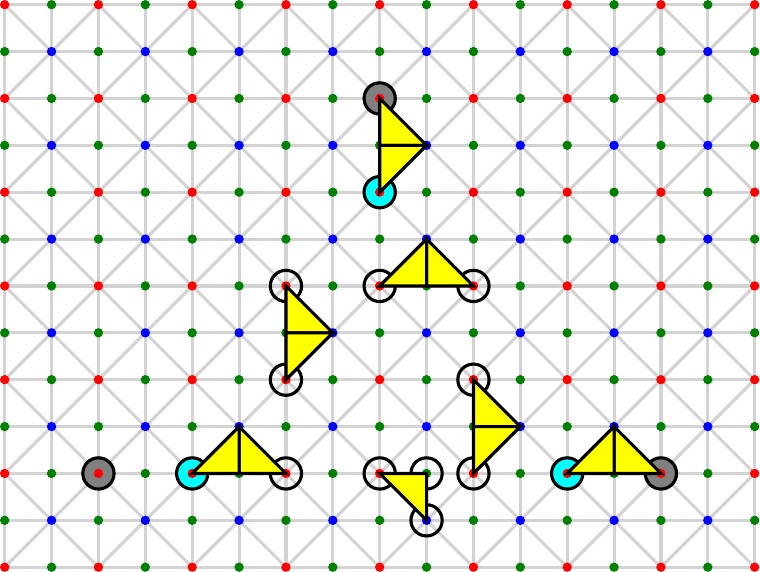}
      \caption{\label{fig:4-8-8-xor-TF}One input TRUE, one FALSE.}
  \end{subfigure}
  \hspace{0.03\columnwidth}
  \begin{subfigure}{0.3\textwidth}
  \includegraphics[width=\columnwidth]{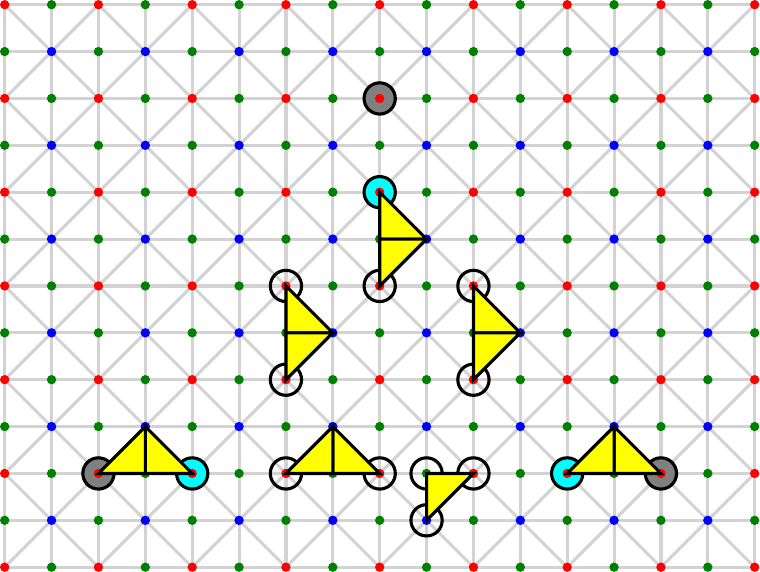}
      \caption{\label{fig:4-8-8-xor-FT}One input TRUE, one FALSE.}
  \end{subfigure}
  \caption{\label{fig:4-8-8-xor-gadget} \textbf{The XOR-gadget for the $\mathbf{4.8.8}$
    colour code}. This is very similar to the $6.6.6$ version, but with
    the central `triangle' blown up (as the $4.8.8$ code does not have a
    cycle that short). More importantly, the shrunk lattice for the
    $4.8.8$ colour code does not have odd cycles, so we need to include a
    NOT-subgadget (the blue/green defects along the bottom
    edge).}
\end{figure*}

\begin{figure*}[t]
  \def\figbasename{4-8-8-rgb-duplicator}
  \begin{subfigure}{0.3\textwidth}
  \includegraphics[width=\columnwidth]{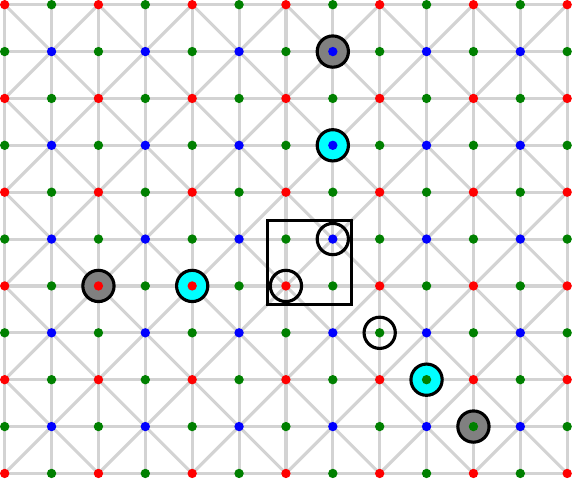}
      \caption{\label{fig:4-8-8-RGB-empty}Empty RGB-duplicator}
  \end{subfigure}
  \hspace{0.03\columnwidth}  
  \begin{subfigure}{0.3\textwidth}
  \includegraphics[width=\columnwidth]{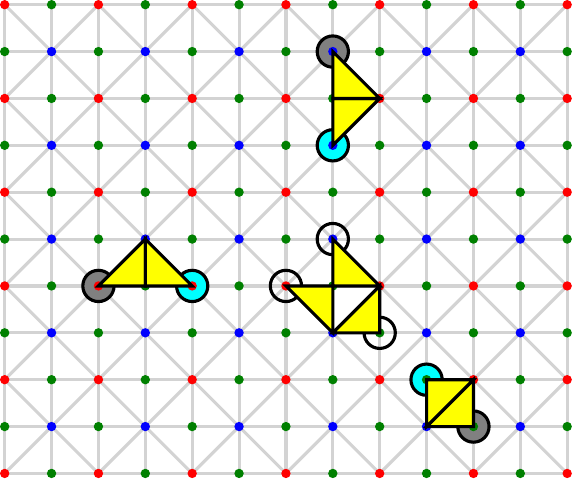}
      \caption{\label{fig:4-8-8-RGB-true}All links TRUE}
  \end{subfigure}
  \hspace{0.03\columnwidth}
  \begin{subfigure}{0.3\textwidth}
  \includegraphics[width=\columnwidth]{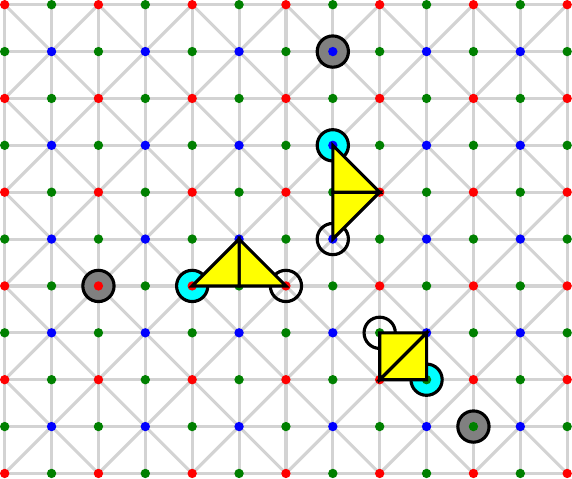}
      \caption{\label{fig:4-8-8-RGB-false}All links FALSE.}
  \end{subfigure}
  \caption{\label{fig:4-8-8-RGB-duplicator} \textbf{The RGB-duplicator
      gadget for the $4.8.8$ colour code}. Using this gadget it is
    easy to construct a double duplicator and then a full variable
    gadget. \textbf{(a)} shows the gadget. Note the syndrome is not
    separated but the box shown is a `syndrome box' -- at least two
    errors must meet the box, so we still have the notion of an exact
    cover. \textbf{(b)} and \textbf{(c)} show exact covers with all
    link nodes TRUE, and all link nodes FALSE respectively. We remark
    that, unlike in the $6.6.6$ colour code, there are other exact
    covers, but the standard parity based argument shows that, in any
    exact cover, all the links have to be in the same state.}
\end{figure*}

\begin{figure}[t]
  \def\figbasename{4-8-8-green-wire}
  \begin{subfigure}{0.47\columnwidth}
  \includegraphics[width=\columnwidth]{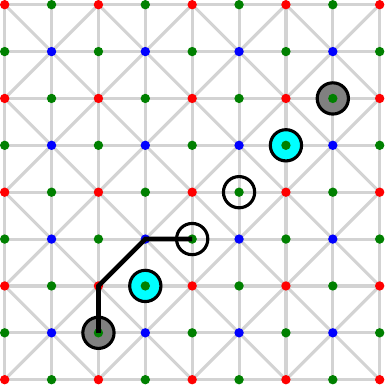}
      \caption{\label{fig:4-8-8-green-wire-extended}Input TRUE, output FALSE\newline }
  \end{subfigure}
  \hspace{0.03\columnwidth}
  \begin{subfigure}{0.47\columnwidth}
  \includegraphics[width=\columnwidth]{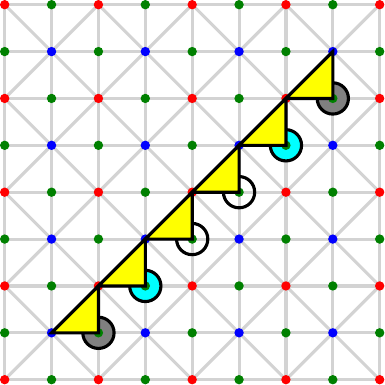}
      \caption{\label{fig:4-8-8-green-wire-true}Input TRUE, output FALSE\newline }
  \end{subfigure}
  \caption{\label{fig:4-8-8-green-wire} \textbf{Anomalous behaviour in
      the $\mathbf{4.8.8}$ color code in the `4' colour}.  Wires and
    link nodes constructed in the `4' colour have some strange
    behaviour in the $4.8.8$ colour code which is why we work with
    wires in one of the `8' colours. Both \textbf{(a)} and
    \textbf{(b)} show a wire like set of defects in the `green' colour
    (which corresponds to the `4' in the 4.8.8 code). In \textbf{(a)},
    we see that the `link' node at the bottom left does not satisfy
    our requirements for a link node as the partner link node is
    distance three from a non-link node in the wire -- a path of
    length three is shown in the figure.  \textbf{(b)} shows a long
    chain of error cancelling except at the ends -- this is precisely
    the kind of behaviour link nodes were designed to prevent. This
    type of behaviour cannot occur in the other colours in the 4.8.8
    code nor in the 6.6.6 colour code.  }
\end{figure}

\section{The details for the other colour codes}\label{s:other-colour-code-details}
In this section we go into detail about the changes we need for other
(planar) colour codes and the 4.8.8 colour code in particular. The
first is a minor tweak of the definition of a wire. In
Section~\ref{ss:wire} we required that `consecutive defects' in a wire
be at distance two. In fact what we really needed is that the two
defects can be \emph{matched} (in the sense of
Figure~\ref{fig:2-errors}) -- that is there exist a pair of errors
giving rise to exactly these two defects with the other nodes in these
two errors cancelling. In the normal hexagonal colour code these two
definitions are equivalent, but that is not the case in other colour
codes.

The shrunk lattice~\cite{bombinTopologicalQuantumDistillation2006} is
a convenient way of capturing this information. It is formed from the
dual lattice, by taking just the vertices of one colour, and joining
two of them exactly if they can be matched in the sense above. Note
there is one shrunk lattice for each colour, and these need not be
isomorphic. For the hexagonal colour code each shrunk lattice is just
a hexagonal lattice; for the $4.8.8$ colour code all three shrunk
lattices are isomorphic to the standard square lattice $\Z^2$. In this
terminology, wires consist of `paths' of defects where consecutive
defects are at distance one in the shrunk lattice.

One surprising extra technicality is that of the link nodes which
require extra care in colour codes other than the hexagonal colour
code.  For example, consider the shrunk lattice for the $4.8.8$ colour
code corresponding to the degree 4 nodes (which all have the same
colour). Then, if we take three nodes $u,v,w$ in this shrunk lattice
with both $uv$ and $vw$ being edges, we find that the distance from
$u$ to $w$ in the \emph{dual} lattice is at most three, not the
distance four we required in our definition of link nodes. In other
words link nodes do not exist in this colour class in the $4.8.8$
code.  Indeed, this anomaly extends to wires: if we make wires in the
colour corresponding to the degree 4 nodes in the dual lattice for the
$4.8.8$ code we find that they can be covered by an extended chain of
errors -- see Figure~\ref{fig:4-8-8-green-wire-extended}. However,
both link nodes and wires behave as expected in either of the other
two colours in the dual lattice, so this is just a surprising
difference rather than a crucial obstruction.

The hexagonal lattice is not bipartite, which corresponds to the fact
that any two red defects in the hexagonal colour code can be connected
by a wire of an even length.  However, $\Z^2$ \emph{is} bipartite and
this is exactly the problem we mentioned in Section~\ref{s:combining}
and~\ref{s:other-colour-codes} -- we cannot always find an \emph{even}
length wire between two red defects. To avoid this issue we introduce
the NOT-subgadget shown in Figure~\ref{fig:4-8-8-not-gadget}. This
gadget does not have a separated syndrome -- we define a cover of this
gadget to be exact if the number of \emph{red} defects its error
cluster meets equals the number of errors (in other words we ignore
the green and blue defect inside the NOT-subgadget in our counting).

We need to use the NOT-subgadget internally in the
XOR-subgadget (which is used to build the garbage collection gadget)
and this is shown in Figure~\ref{fig:4-8-8-xor-gadget}

The RGB-duplicator subgadget is shown in
Figure~\ref{fig:4-8-8-RGB-duplicator}. We remark that, unlike in the
6.6.6 code, the syndrome is not separated. However, we have a
syndrome-box shown in Figure~\ref{fig:4-8-8-RGB-empty}, and it is easy
to see that, in any cover, at least two errors must meet the syndrome
box. Since all other defects are at distance at least two from the
syndrome box this means that we still have the notion of an exact
cover.

The crossing gadget also appears to exist. However, it is
constructed out of two multi-coloured-crossing subgadgets: one crossing red and blue, and
one crossing red and green. In most cases (the $6.6.6$ colour code is an exception) this will necessitate constructing two different multi-coloured crossing subgadgets. Again, as the crossing gadget is only needed for
code capacity noise we omit these constructions.

The final regular planar colour code is the $4.6.12$ colour code. We
would expect our results to hold in this case but since this code is
much less practical for QEC (the weight twelve stabilisers are
problematic) we have not verified the details.

\section{Approximating the minimum weight decoding (mathematical details)}\label{a:approx}

In this section we discuss the question of approximating the minimum
weight decoding in more mathematical detail. We point the interested
reader to two concepts in the literature: FPTAS and PTAS. As we said
in Section~\ref{s:approx} some NP-hard problems have polynomial time algorithms that
approximate the optimal solution arbitrarily closely. The following
definition makes this more precise.

\begin{defn}
  An optimisation problem has a \emph{polynomial time approximation scheme
  (PTAS)}, if for all $\eps>0$ there exists a polynomial time algorithm
  that finds a solution within a factor $1+\eps$ of the optimum.
\end{defn}

Some NP-hard problems have a PTAS (e.g., Euclidean TSP) and some do
not (e.g., Max-3SAT).  It is important that the polynomial is allowed
to depend \emph{arbitrarily} on $\eps$. The problem is said to have a
Fully Polynomial Time Approximation Scheme (FPTAS) if the run time is
bounded by a polynomial in both the size of the problem and
$1/\eps$. Assuming that $\text{P}\not=\text{NP}$, it is known that no NP-hard problem has an FPTAS. In other
words, for any NP-hard problem, including our problem of minimum
weight decoding for the colour code, as we try and approximate more
closely the run-time must behave super-polynomially in $1/\eps$ --
i.e., it behaves badly as $\eps$ gets small. However, how close a
practical algorithm can get remains an open question.

\vfil

\bibliography{articles_tidied.bib} 

\end{document}